\documentclass[a4paper,11pt]{article}

\usepackage{fullpage}
\usepackage{times}
\usepackage{soul}
\usepackage{url}
\usepackage[utf8]{inputenc}
\usepackage{subcaption}
\usepackage{multirow}
\usepackage{graphicx}
\usepackage{amsmath}
\usepackage{booktabs}

\usepackage[ruled,vlined]{algorithm2e}
\SetArgSty{textrm}

\usepackage{enumitem}

\usepackage{libertine}

\usepackage{amssymb,amsmath,amsfonts,amstext,amsthm}

\usepackage[breaklinks=true]{hyperref}
\usepackage[svgnames]{xcolor}
\usepackage[capitalise,nameinlink]{cleveref}
\hypersetup{colorlinks={true},linkcolor={DarkBlue},citecolor=[named]{DarkGreen}}
\crefname{enumi}{}{}

\usepackage{tikz}  
\usetikzlibrary{arrows}
\usetikzlibrary{patterns,snakes}
\usetikzlibrary{decorations.shapes}
\tikzstyle{overbrace text style}=[font=\tiny, above, pos=.5, yshift=5pt]
\tikzstyle{overbrace style}=[decorate,decoration={brace,raise=5pt,amplitude=3pt}]
\usetikzlibrary{shapes.geometric}
\usetikzlibrary{shapes,positioning}
\usetikzlibrary{calc,positioning}
\usetikzlibrary{shapes.multipart}
\usetikzlibrary {graphs}

\definecolor{cadmiumgreen}{rgb}{0.0, 0.42, 0.24}

\usepackage{bbm}

\newtheorem{theorem}{Theorem}[section]
\newtheorem{corollary}[theorem]{Corollary}

\newtheorem{lemma}[theorem]{Lemma}

\newtheorem{claim}[theorem]{Claim}

\theoremstyle{definition}

\newtheorem{remark}[theorem]{Remark}

\newtheorem{observation}[theorem]{Observation}

\usepackage[sort]{natbib}%

\usepackage{mathtools}
\usepackage{xcolor} 
\usepackage{tcolorbox}

\usepackage{authblk}

\usepackage{subcaption}

\newcommand{\SC}{\text{\normalfont SC}}
\newcommand{\MC}{\text{\normalfont MC}}

\newcommand{\aris}[1]{{\color{blue}[Aris: #1]}}

\newcommand{\tas}{\textrm{TAS}}
\newcommand{\ord}{\textrm{ORD}}
\newcommand{\dis}{\textrm{DIST}}

\begin{document}

\allowdisplaybreaks

\title{\bf Improved Metric Distortion via Threshold Approvals}

\author[1]{Elliot Anshelevich}
\author[2]{Aris Filos-Ratsikas}
\author[1]{Christopher Jerrett}
\author[3]{Alexandros A. Voudouris}

\affil[1]{Department of Computer Science, Rensselaer Polytechnic Institute, USA}
\affil[2]{School of Informatics, University of Edinburgh, UK}
\affil[3]{School of Computer Science and Electronic Engineering, University of Essex, UK}

\renewcommand\Authands{ and }
\date{}

\maketitle

\begin{abstract}
We consider a social choice setting in which agents and alternatives are represented by points in a metric space, and the cost of an agent for an alternative is the distance between the corresponding points in the space. The goal is to choose a single alternative to (approximately) minimize the social cost (cost of all agents) or the maximum cost of any agent, when only limited information about the preferences of the agents is given. Previous work has shown that the best possible distortion one can hope to achieve is $3$ when access to the ordinal preferences of the agents is given, even when the distances between alternatives in the metric space are known. We improve upon this bound of $3$ by designing deterministic mechanisms that exploit a bit of cardinal information. We show that it is possible to achieve distortion $1+\sqrt{2}$ by using the ordinal preferences of the agents, the distances between alternatives, and a threshold approval set per agent that contains all alternatives for whom her cost is within an appropriately chosen factor of her cost for her most-preferred alternative. We show that this bound is the best possible for any deterministic mechanism in general metric spaces, and also provide improved bounds for the fundamental case of a line metric. 
\end{abstract}

\section{Introduction}
\textit{Social choice theory} is concerned with aggregating the heterogeneous preferences of individuals over a set of outcomes into a single decision \citep{brandt2016handbook}. Besides its many applications in traditional domains, such as political elections or voting for policy issues, social choice theory has also been at the epicenter of areas such as \textit{multi-agent systems} \citep{ephrati1991clarke}, \textit{recommendation systems }\citep{ghosh1999voting}, \textit{search engines }\citep{dwork2001rank}, and \textit{crowdsourcing} applications \citep{mao2012social}. Aggregation methods inspired from social choice theory have also been used in relation to machine learning, such as for \textit{regression} and \textit{estimation} tasks \citep{caragiannis2016truthful,chen2018strategyproof,chen2020truthful,kahng2020strategyproof}, as well as \textit{virtual democracy} \citep{noothigattu2018voting,kahng2019statistical,peters2020explainable}.

A central theme underpinning many of these applications is the interplay between the amount of available information and the efficiency of the implemented decisions \citep{xia2022group,zhao2019learning,mandal2019thrifty,amanatidisdon,boutilier2015optimal}. Indeed, it is often the case that access to the preferences of the participants (or \emph{agents}) is restricted due to lack of enough statistical data, or due to computational or cognitive limitations in the elicitation process. This leads to the natural question of whether we can design aggregation rules (or \emph{mechanisms}) that achieve high efficiency even when presented with informational restrictions.

This is the main question studied in the literature of \emph{distortion} \citep{procaccia2006distortion,anshelevich2021distortion}, which measures the deterioration of an aggregate social objective due to limited information about the preferences of the agents. In its original setup (e.g., see \citep{boutilier2015optimal}), limited information is interpreted as having access only to the \emph{ordinal preference rankings} over the possible outcomes, instead of a complete \emph{cardinal utility} structure that fully captures the intensity of the preferences. Over the years, the notion of distortion has been refined to capture different kinds of informational limitations, including restricted ordinal information \citep{gross2017vote,kempe2020communication}, communication complexity \citep{mandal2019thrifty,mandal2020optimal}, and query complexity \citep{amanatidis2021peeking,amanatidisdon}.

The literature on \emph{metric distortion} considers scenarios where the agents and the outcomes exist in a (possibly high-dimensional) metric space, and the costs (rather than utilities) are given by distances that satisfy the triangle inequality. The economic interpretation of these distances is typically in terms of the proximity for political or ideological issues along different axes (e.g., liberal to conservative, or libertarian to authoritarian). The metric distortion of social choice mechanisms, firstly studied in the works of \citet{anshelevich2018approximating} and \citet{anshelevich2017randomized}, is one of the most well-studied settings in the context of this literature, with many variants of the main setting being considered in recent years (see Section 3 of the survey of \citet{anshelevich2021distortion} for an overview). Since its inception in 2015 (the conference version of \citep{anshelevich2018approximating}), the ``holy grail'' of this literature was a mechanism with a metric distortion of $3$ for the \emph{social cost objective} (i.e., the sum of costs of all agents), which would match the corresponding lower bound shown by \citet{anshelevich2018approximating}. The upper bound was finally established by \citet{gkatzelis2020resolving} via the \textsc{Plurality Matching} mechanism, and then later also by \citet{pluralityveto2022} via the much simpler \textsc{Plurality Veto} mechanism. For the other most natural objective, that of minimizing the \emph{maximum cost} of any agent, the best possible bound can also easily be seen to be $3$. 

Crucially, these tight bounds of $3$ are achieved by mechanisms without access to any information of a cardinal nature, i.e., they work by using only the ordinal preference rankings of the agents. It can also be observed (e.g., see \citep{anshelevich2021ordinal}) that if one also assumes access to the distances between the outcomes in the metric space, then $3$ is still the best bound one can hope for. At the same time, a growing recent literature advocates to take the natural next step, and study the distortion when a small amount of cardinal information about the preferences is also available, which, in many cases, is reasonable to elicit \citep{abramowitz2019awareness,amanatidis2021peeking,amanatidis2022few,amanatidisdon,benade2021preference,bhaskar2018truthful}. This motivates the following question:

\begin{quote}
\emph{
    Is it possible to beat the barrier of \ $3$ on the metric distortion if we also have access to limited information of a cardinal nature?}
\end{quote}

\subsection{Our Contribution}
We consider a setting in which a set $N$ of agents and a set $A$ of alternatives are positioned on a metric space, and the preferences of the agents are captured by the distances between their positions and that of the different alternatives. We are interested in the distortion of deterministic mechanisms in terms of the social cost and the maximum cost objectives, i.e., the ratio of the cost of the alternative chosen by the mechanism over the smallest possible cost over all alternatives, taken over all possible instances. Note that with full cardinal information, achieving an optimal distortion of $1$ is trivial.

We answer the main question posed above in the affirmative, by assuming access to the ordinal preferences, the distances between the alternatives in the metric space, and some additional limited cardinal information about the costs of the agents. In particular, for each agent $i \in N$, we have access to an \emph{$\alpha$-threshold approval set} ($\alpha$-TAS) $A_i$, which contains all the alternatives for which the agent has cost at most a factor $\alpha$ times the cost for her most-preferred alternative, with $\alpha$ being a value that can be chosen by the mechanism designer. For general metric spaces, we design mechanisms that achieve a distortion of $1+\sqrt{2}$ in terms of the social cost and the maximum cost objectives, thus beating the aforementioned barrier of $3$. We complement these results with lower bounds on the distortion that \emph{any} deterministic mechanism can achieve when using this amount of information.

For the fundamental special case of a line metric, we provide refined tight bounds on the distortion for both objectives, showing an even larger separation from the bound of $3$ that holds in the absence of any cardinal information and even when the metric is a line. More generally, we fine-tune our analysis to the presence/absence of all three types of information (i.e., only $\alpha$-TAS, $\alpha$-TAS + ordinal preferences, $\alpha$-TAS + known alternative distances, or $\alpha$-TAS + ordinal preferences + known alternative distances) and show that, in most cases, the best possible distortion bounds can be obtained even when only using two types of information. Our results are summarized in Table~\ref{tab:overview}.

\renewcommand{\arraystretch}{1.3}
\begin{table}[t]
\centering
\begin{tabular}{c|c|c|c|c|c}
                                &           & $\tas\cap\ord\cap\dis$ & $\tas\cap\ord$ & $\tas\cap\dis$ & $\tas$ \\
\hline
\multirow{2}{*}{Social cost}    & General   & $[2,1+\sqrt{2}]$  & $1+\sqrt{2}$  & $[1+\sqrt{2},3]$ &  \multirow{2}{*}{$\Theta(n)$} \\
\cline{3-5}
                                & Line      & $2\sqrt{2}-1$     & $2$           & $2\sqrt{2}-1$ \\
\hline
\multirow{2}{*}{Max cost}       & General   & $1+\sqrt{2}$      & \multicolumn{2}{c|}{$[1+\sqrt{2},3]$} & \multirow{2}{*}{$3$}\\
\cline{3-5}
                                & Line      & \multicolumn{3}{c|}{$1+\sqrt{2}$} \\
\hline
\end{tabular}
\caption{An overview of our distortion bounds for deterministic mechanisms that use different types of information. $\tas$ is used to refer to the class of mechanism that have access to the ordinal preferences of the agents, $\dis$ is used for mechanisms that have access to the distances between alternatives, and $\tas$ is used for the mechanism with access to the $\alpha$-threshold approval sets.}
\label{tab:overview}
\end{table}

\subsection{Related Work and Discussion}
The distortion literature is rather extensive, focusing primarily on the utilitarian normalized setting (e.g., see \citep{boutilier2015optimal,caragiannis2017subset,caragiannis2022truthful,FMV20,ebadian2022optimized,filos2014truthful,mandal2019thrifty,mandal2020optimal,bhaskar2018truthful,benade2021preference}) or the metric setting that we consider in this paper (e.g., see \citep{abramowitz2017utilitarians,abramowitz2019awareness,anshelevich2016blind,anshelevich2016truthful,anshelevich2017randomized,anshelevich2017tradeoffs,anshelevich2018approximating,anshelevich2021ordinal,anshelevich2022distortion,CSV22,caragiannis2022truthful2,fain2019random,feldman2016voting,FV21,FKVZ23,ghodsi2019abstention,gkatzelis2020resolving,pluralityveto2022,munagala2019improved,V23}). For more details, we refer the reader to the recent survey of \citet{anshelevich2021distortion}. Below we discuss the works that are mostly relevant to the setting that we study here. 

The term ``approval threshold'' (in fact, ``approval threshold vote'') was introduced in the conference version of the paper of \citet{benade2021preference} (and was later also used by \citet{bhaskar2018truthful}) to refer to an elicitation device that returns a set of alternatives that an agent values higher than a given threshold. \citet{amanatidis2021peeking} introduced a setting where limited cardinal information is elicited via a set of \emph{value} or \emph{comparison queries}; an agent is asked to provide her value for a given alternative, or indicate whether she prefers one alternative over another by more than a given factor. An approval threshold in their setting can be constructed by a single value query (for the agent's most-preferred alternative), and then a number of comparison queries that is logarithmic in the number of alternatives. 

The $\alpha$-TAS that we consider are closely related to the comparison query model of \citet{amanatidis2021peeking}, in that they do not encode information about any absolute value for the cost of the agents, but rather only information that is \emph{relative} to their cost for their most-preferred alternatives. In this sense, these sets are \emph{relative}, and can be viewed as the outcomes of a single relative threshold query, or, equivalently, the outcome of logarithmically-many comparison queries. From a cognitive standpoint, this is an even more conceivable elicitation device since the agents are only asked to perform (cardinal) comparisons rather than to come up with absolute cost numbers. Such comparisons are motivated by and rooted in the ideas of the Von Neumann-Morgenstern utility theory \citep{von2007theory}; see also \citep{amanatidis2021peeking}. 

A common characteristic of all aforementioned works is that they do not consider the distortion in the metric setting. \citet{abramowitz2019awareness} studied the metric distortion of mechanisms that use limited cardinal information, namely information about whether the ratio between the costs of an agent for any two candidates is larger or smaller than a chosen threshold $\tau$. Their elicitation method is different from ours, and requires information about the relative costs for all pairs of alternatives. In contrast, our $\alpha$-TAS only require aggregate information about the set $A_i$ of alternatives. The two settings coincide only in the case where there are just $2$ alternatives overall, which allows us to obtain a lower bound of $2\sqrt{2}-1 \approx 1.83$ for the social cost for the line metric from their work, for which we also show a matching upper bound. 

\section{Preliminaries}
We consider a single-winner social choice setting with a set $N$ of $n$ {\em agents} and a set $A$ of $m$ {\em alternatives}. Agents and alternatives are represented by points in a metric space. For any $x,y \in N \cup A$, let $d(x,y)$ denote the {\em distance} between the points representing $x$ and $y$. The distances satisfy the standard conditions for metric spaces, namely that $d(x,y) = d(y,x)$, $d(x,y)=0$ if $x=y$, and $d(x,y) \leq d(x,z) + d(z,y)$; the last condition is called the \emph{triangle inequality}. We will use the tuple $I=(N,A,d)$ to denote an \emph{instance} of our setting, and will use $\mathcal{I}$ to denote the set of all instances. 

\paragraph{Input information.} 
A (deterministic) \emph{mechanism} $M$ takes as input some information related to the distances between agents and alternatives, and outputs a single alternative as the {\em winner}. 
We will consider combinations of the following three types of information:
\begin{enumerate}[label=Inf.\arabic*]
    \item The \emph{ordinal preferences} $\left(\succ_i\right)_{i \in N}$ of the agents which are induced by the distances, where $\succ_i$ is a complete ordering over the elements of $A$, and $x \succ_i y$ implies that $d(i,x) \leq d(i,y)$ for any $i \in N$ and $x,y \in A$. \label{enum1}
    
    \item The distances $d(x,y)$ between any pair of alternatives $x,y \in A$. \label{enum2}
    
    \item A set of {\em $\alpha$-threshold approval sets} ($\alpha$-TAS) $(A_i)_{i \in N}$ such that $d(i,x) \leq \alpha \cdot d(i,o_i)$ for any $x \in A_i$, where $o_i$ is the most-preferred alternative of agent $i$ (i.e., $o_i \succ_i y$ for any $y \in A$). We will say that agent $i$ {\em approves} alternative $x$ in case $x \in A_i$. \label{enum3}
\end{enumerate}
Clearly, all of \Cref{enum1,enum2,enum3} can be inferred from an instance $I$. We will use the term $\text{Inf}(I)$ to denote the information that is available to a mechanism for instance $I$. For ease of reference, we classify mechanisms by the type(s) of information that they use at inputs. In particular, we will let $\ord$, $\dis$, and $\tas$ denote the classes of mechanisms that have access to the ordinal preferences of the agents (\Cref{enum1}), the distances between alternatives (\Cref{enum2}), and the $\alpha$-TAS (\Cref{enum3}), respectively. If a mechanism has access to more than one types of information, then it is at the intersection of the corresponding classes. For example, a mechanism in $\ord \cap \tas$ has access to the ordinal preferences and the $\alpha$-TAS, but not to the distances between alternatives.  

As we already explained in the Introduction, the bulk of the metric distortion literature considers mechanisms that are only in $\ord$, which results in a best-possible distortion of $3$ \citep{anshelevich2018approximating,gkatzelis2020resolving,pluralityveto2022} (see below for the formal definition). For mechanisms in $\ord \cap \dis$, the best possible distortion is still $3$ \citep{anshelevich2021ordinal}. Our main results are for the case of having access to all types of information (i.e., mechanisms in $\ord \cap \dis \cap \tas$), but we also provide results for mechanisms in $\ord \cap \tas$ and $\dis \cap \tas$ independently, as well as $\tas$ in isolation. 

\paragraph{Objectives and Distortion.} 
We will consider the two well-known minimization objectives, the \emph{social cost} $\SC$ and the {\em maximum cost} $\MC$, which for an alternative $x$ are defined as $\SC(x) = \sum_{i\in N}d(i,x)$ and $\MC(x) = \max_{i \in N} d(i,x)$.
Given an objective $\text{F} \in \{\SC,\MC\}$, the {\em distortion} of a mechanism $M$ for $\text{F}$ is defined as 
\begin{align*}
\sup_{I \in \mathcal{I}} \frac{\text{F}(M(\text{Inf}(I))|I)}{\min_{j \in A} \text{F}(j|I)},    
\end{align*}
where $M(\text{Inf}(I))$ denotes the alternative chosen by the mechanism when given as input the information $\text{Inf}(I)$ for instance $I$. Clearly, the distortion is at least $1$ for any mechanism, and our goal is to determine the best possible distortion when given combinations of \Cref{enum1,enum2,enum3}.


\section{Social Cost}
In this section, we show bounds on the distortion for the social cost objective for mechanisms in $\ord \cap \dis \cap \tas$. Our main result is a mechanism, coined $(1+\sqrt{2})$-\textsc{Minisum-TAS-Distance}, which achieves a distortion of $1+\sqrt{2} \approx 2.42$, thus breaking the barrier of $3$, which is the best possible without access to the $\alpha$-TAS. In fact, this mechanism does not require access to the ordinal preferences of the agents, i.e., it is in $\dis \cap \tas$. We complement this result with a lower bound of $2$ on the distortion of any deterministic mechanism $M \in \ord \cap \dis \cap \tas$. 

For mechanisms in $\ord \cap \tas$ and $\dis \cap \tas$, we provide stronger lower bounds of $1+\sqrt{2}$; this establishes that $(1+\sqrt{2})$-\textsc{Minisum-TAS-Distance} is the best possible among mechanisms in $\dis \cap \tas$. We also prove that the distortion of any mechanism in $\tas$ is bound to be bad, namely $\Theta(n)$. For the special case of the line metric, we provide more refined \emph{tight} bounds depending on the amount of information available in the input. 


\subsection{Results for General Metric Spaces}\label{sec:sc-general}
We start by showing an upper bound of $1+\sqrt{2}$ for general metric spaces using a mechanism to which we refer as {\sc $\alpha$-Minisum-TAS-Distance}. This mechanism is quite intuitive and does not require any information about the ordinal preferences, i.e., it is in the class $\dis \cap \tas$. The mechanism chooses as the winner an alternative that minimizes the sum of distances from the $\alpha$-TAS of the agents, where the distance between an alternative $x$ and a set $A_i$ is defined as the distance of $x$ from its the closest alternative in $A_i$. 
See Mechanism~\ref{mech:minisum-tas-distance} for a description using pseudocode. 

\begin{algorithm}[t]
\SetNoFillComment
\caption{\sc $\alpha$-Minisum-TAS-Distance}
\label{mech:minisum-tas-distance}
{\bf Input:} Distances between alternatives, $\alpha$-TAS\;
{\bf Output:} Winner $w$\;
$w \in \arg\min_{x \in A} \bigg\{ \sum_{i \in N} \min_{j \in A_i} d(j,x) \bigg\}$\;
\end{algorithm}


\begin{theorem}\label{thm:sc-general-upper}
In general metric spaces, the distortion of {\sc $\alpha$-Minisum-TAS-Distance} for the social cost objective is at most $\max\left\{\alpha,2+\frac{1}{\alpha}\right\}$.
\end{theorem}

\begin{proof}
Recall that $o_i$ is the most-preferred alternative of agent $i$; hence, $d(i,j) \leq \alpha \cdot d(i,o_i)$ for any $j \in A_i$.
Let $o$ be the optimal alternative (i.e., the alternative with the smallest social cost) and let $w$ be the alternative chosen by the mechanism. 
For any agent $i$, denote by $x_i \in A_i$ the alternative that is closest to $w$, and by $y_i \in A_i$ the alternative that is closest to $o$. By the definition of $w$, we have that $\sum_{i \in N} d(x_i,w) \leq \sum_{i \in N} d(y_i,o)$.
Using the triangle inequality, we now have:
\begin{align*}
    \SC(w) = \sum_{i \in N} d(i,w) 
        &\leq \sum_{i \in N} d(i,x_i) + \sum_{i \in N} d(x_i,w) \\
        &\leq \sum_{i \in N} d(i,x_i) + \sum_{i \in N} d(y_i,o).
\end{align*}
We make the following observations:
\begin{itemize}
    \item[-] For every agent $i$ such that $o \in A_i$, $y_i=o$, and thus $d(y_i,o)=0$, as well as $d(i,o_i) \leq d(i,o)$.
    \item[-] For every agent $i$ such that $o \not\in A_i$, $d(y_i,o) \leq d(o_i,o)$ and $d(i,o_i) < \frac{1}{\alpha}\cdot d(i,o)$. 
\end{itemize}
Combining the above with the fact that $d(i,x_i) \leq \alpha \cdot d(i,o_i)$ and using the triangle inequality, we obtain
\begin{align*}
    \SC(w) 
        &\leq \sum_{i \in N} d(i,x_i) + \sum_{i \in N} d(y_i,o) \\
        &= \sum_{i: o \in A_i} \bigg( d(i,x_i) + d(y_i,o) \bigg) + \sum_{i: o \not\in A_i} \bigg( d(i,x_i) + d(y_i,o) \bigg) \\
        &\leq \sum_{i: o \in A_i}  d(i,x_i) + \sum_{i: o \not\in A_i} \bigg( d(i,x_i) + d(o_i,o) \bigg) \\
        &\leq \alpha \cdot \sum_{i: o \in A_i}  d(i,o_i) + \sum_{i: o \not\in A_i} \bigg( \alpha \cdot d(i,o_i) + d(i,o_i) + d(i,o) \bigg) \\
        &\leq \alpha \cdot \sum_{i: o \in A_i}  d(i,o) + \bigg(2 + \frac{1}{\alpha} \bigg) \cdot \sum_{i: o \not\in A_i} d(i,o) \\
        &\leq \max\left\{\alpha, 2 + \frac{1}{\alpha}\right\} \cdot \SC(o). 
\end{align*}
Consequently, the distortion is at most $\max\left\{\alpha, 2 + \frac{1}{\alpha}\right\}$.
\end{proof}

By optimizing over $\alpha$, we obtain the following corollary. 

\begin{corollary}
In general metric spaces, the distortion of {\sc $(1+\sqrt{2})$-Minisum-TAS-Distance} for the social cost objective is at most $1+\sqrt{2}$.
\end{corollary}

We complement the upper bound of $1+\sqrt{2}$ by a nearly tight lower bound of $2$ for mechanisms that use all three types of information. 
We first show a lower bound of $\min\{3,\alpha\}$ on the distortion of such mechanisms; the following lemma will be used in several lower bound proofs, and also suggests that to beat the bound of $3$, a mechanism must use $\alpha < 3$.  

\begin{lemma}\label{lem:sc-general-alpha-all-three-types}
For the social cost objective, the distortion of any mechanism $M \in \ord \cap \dis \cap \tas$ is at least $\min\{3,\alpha\}$, for any $\alpha \geq 1$.
\end{lemma}

\begin{proof}
We consider an instance with $n$ agents and $n$ alternatives $A=\{a_1,\ldots,a_n\}$. The distance between any two alternatives is $2$. The agents have cyclic ordinal preferences over the alternatives such that for agent $i$:
\begin{align*}
    a_i \succ_i \ldots \succ_i a_n \succ_i a_1 \succ \ldots \succ_i a_{i-1}.
\end{align*}
In addition, all agents approve all alternatives, that is, $a \in A_i$ for every $a \in A$ and every agent $i$. 
Since the instance is completely symmetric, we can assume without loss of generality that the mechanism chooses alternative $a_i$ for some $i \in [n]$ as the winner when given as input this information about the instance. 

\begin{figure}[t]
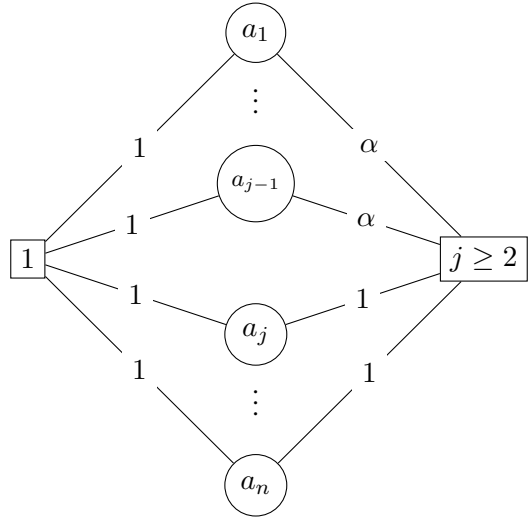

\centering
\tikz {
  \node (a1) [circle, draw] at (0,3) {$a_1$};
  \node (dots) at (0,2.2) {$\vdots$};
  \node (ai-1) [circle, draw] at (0,1)  {\footnotesize $a_{j-1}$};
  \node (ai) [circle, draw] at (0,-1)  {$a_j$};
  \node (dots) at (0,-1.8) {$\vdots$};
  \node (an) [circle, draw] at (0,-3)  {$a_n$};
  
  \node (i) [rectangle, draw] at (-3,0) {$1$};
  \node (j) [rectangle, draw]  at (3,0) {$j\geq 2$};
  
  \draw (i) edge node [midway, fill=white] {$1$} (a1) ;
  \draw (i) edge node [midway, fill=white] {$1$} (ai-1) ;
  \draw (i) edge node [midway, fill=white] {$1$} (ai) ;
  \draw (i) edge node [midway, fill=white] {$1$} (an) ;
  
  \draw (j) edge node [midway, fill=white] {$\alpha$} (a1) ;
  \draw (j) edge node [midway, fill=white] {$\alpha$} (ai-1) ;
  \draw (j) edge node [midway, fill=white] {$1$} (ai) ;
  \draw (j) edge node [midway, fill=white] {$1$} (an) ;
}
\caption{An example of the metric space used in the proof of Lemma~\ref{lem:sc-general-alpha-all-three-types} with $i=1$ to show a lower bound of $\alpha$, assuming that $\alpha \leq 3$. Agents are represented by rectangles, and alternatives by circles.}
\label{fig:thm:ordinal-list-lower:alpha}
\end{figure}

Observe that, besides agent $i$, all other $n-1$ agents prefer $a_{i-1}$ to $a_i$, interpreting $a_0$ as $a_n$. Given this, we aim to define a metric space that is consistent with the distances between the alternatives, as well as with the ordinal preferences of the agents, so as to maximize the social cost of $a_i$ while minimizing the social cost of $a_{i-1}$. In particular, we have the following metric space (see Figure~\ref{fig:thm:ordinal-list-lower:alpha} for an example with $i=1$):
\begin{itemize}
    \item[-] Agent $i$ is at distance $1$ from all alternatives.
    \item[-] Agent $j < i$ is at distance $1$ from all alternatives in $\{a_j \ldots, a_{i-1}\}$ and at distance $\min\{3,\alpha\}$ from the remaining alternatives in $\{a_i, \ldots, a_n, a_1 \ldots, a_{j-1}\}$.
    \item[-] Agent $j > i$ is at distance $1$ from the alternatives in $\{a_j, \ldots, a_n\}$ and at distance $\min\{3,\alpha\}$ from the remaining alternatives in $\{a_1, \ldots, a_{j-1}\}$; note that $a_i$ belong to the latter set of alternatives in this case.
    \item[-] The distance between two alternatives is determined by the length of the minimum path that connects them; hence, it is exactly $2$ (due to the distance $1$ of agent $i$ from all alternatives). 
\end{itemize}
Since all agents, beside $i$, are at distance $1$ from $a_{i-1}$ and at distance $min\{3,\alpha\}$ from $a_i$, we have
$$\frac{\SC(a_i)}{\SC(a_{i-1})} = \frac{1 + (n-1)\cdot \min\{3,\alpha\}}{n}.$$
As $n$ becomes arbitrarily large, we obtain a lower bound of $\min\{3,\alpha\}$ on the distortion. 
\end{proof}

We are now ready to show the lower bound of $2$ on the distortion of any mechanism that uses all three types of information. 

\begin{theorem}\label{thm:sc-general-all-three-types}
For the social cost objective, the distortion of any mechanism $M \in \ord \cap \dis \cap \tas$ is at least $2$. 
\end{theorem}

\begin{proof}
Assume by contradiction that there exists mechanism with distortion strictly less than $2$, and consider any such mechanism.  
By Lemma~\ref{lem:sc-general-alpha-all-three-types}, the distortion of any deterministic mechanism is at least $\min\{3,\alpha\}$, even when the ordinal preferences of the agents are given and the distances between alternatives are fixed. So, the mechanism must use $\alpha < 2$.

\begin{figure}[t]
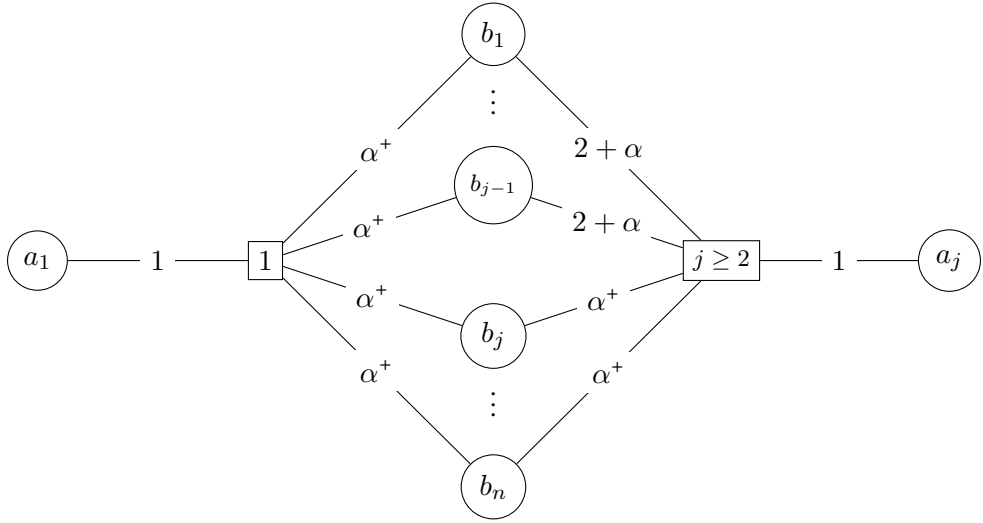

\centering
\tikz {
  \node (b1) [circle, draw] at (0,3) {$b_1$};
  \node (dots) at (0,2.2) {$\vdots$};
  \node (bi-1) [circle, draw] at (0,1)  {\footnotesize $b_{j-1}$};
  \node (bi) [circle, draw] at (0,-1)  {$b_j$};
  \node (dots) at (0,-1.8) {$\vdots$};
  \node (bn) [circle, draw] at (0,-3)  {$b_n$};
  
  \node (i) [rectangle, draw] at (-3,0) {$1$};
  \node (ai) [circle, draw] at (-6,0) {$a_1$};
  \node (j) [rectangle, draw]  at (3,0) {\footnotesize $j\geq 2$};
  \node (aj) [circle, draw] at (6,0) {$a_j$};

  \draw (i) edge node [midway, fill=white] {$1$} (ai) ;
  \draw (j) edge node [midway, fill=white] {$1$} (aj) ;

  \draw (i) edge node [midway, fill=white] {$\alpha^{\text{\footnotesize+}}$} (b1) ; 
  \draw (i) edge node [midway, fill=white] {$\alpha^{\text{\footnotesize+}}$} (bi-1) ; 
  \draw (i) edge node [midway, fill=white] {$\alpha^{\text{\footnotesize+}}$} (bi) ; 
  \draw (i) edge node [midway, fill=white] {$\alpha^{\text{\footnotesize+}}$} (bn) ; 
  
  \draw (j) edge node [midway, fill=white] {$2+\alpha$} (b1) ; 
  \draw (j) edge node [midway, fill=white] {$2+\alpha$} (bi-1) ; 
  \draw (j) edge node [midway, fill=white] {$\alpha^{\text{\footnotesize+}}$} (bi) ; 
  \draw (j) edge node [midway, fill=white] {$\alpha^{\text{\footnotesize+}}$} (bn) ; 

}
\caption{An example of the metric space used in the proof of Theorem~\ref{thm:sc-general-all-three-types} to show a lower bound of $1+2/\alpha$ for $i=1$. Distances between alternatives are not shown, but are assumed to be fixed. Agents are represented by rectangles, and alternatives by circles. The notation $\alpha^{\text{\footnotesize+}}$ is used to represent a value slight larger than $\alpha$ (instead of $\alpha + \delta$). It is not hard to verify that the specified distances satisfy the triangle inequality; for example, the distance of the edge between agent $j$ and alternative $b_1$ is $2+\alpha$, whereas the length of a path from $j$ to $b_1$ via any other $b$-type alternative is $3\alpha^{\text{\footnotesize+}}$, and the length of the path via $a_j$ is again $2+\alpha$. }
\label{fig:thm:sc-general-all-three-types}
\end{figure}

We consider an instance with $n$ agents and $2n$ alternatives $A=\{a_1,\ldots,a_n,b_1,\ldots,b_n\}$. Any pair of $a$-type alternatives are at distance $1+2\alpha$, any pair of $b$-type alternatives are at distance $1+2\alpha$, and any pair of alternatives consisting of one $a$-type and one $b$-type are at distance $2\alpha$.
Let $[X]$ be an arbitrary permutation of the elements of $X$. Each agent $i \in [n]$ has the ordinal preference
\begin{align*}
    a_i \succ_i b_i \succ_i \ldots \succ_i b_n \succ_i b_1 \succ \ldots \succ_i b_{i-1} \succ [\{a_1, \ldots, a_n\} \setminus \{a_i\}].
\end{align*}
In addition, agent $i$ reports that she only approves alternative $a_i$, that is, $A_i = \{a_i\}$.
Let $\delta>0$ be an infinitesimal.
Now, suppose that, for some $i \in [n]$, the mechanism chooses $a_i$ or $b_i$ as the winner. 
We consider the following metric space:
\begin{itemize}
    \item[-] Agent $i$ is 
    at distance $1$ from alternative $a_i$, 
    at distance $\alpha+\delta$ from each alternative in $\{b_1, \ldots, b_n\}$, 
    and at distance $1+2\alpha+2\delta$ from the remaining alternatives in $\{a_1, \ldots, a_n\}\setminus\{a_i\}$.

    \item[-] Agent $j < i$ is at distance $1$ from alternative $a_j$,    
    at distance $\alpha+\delta$ from every alternative in $\{b_j, \ldots, b_{i-1}\}$, 
    and at distance $2+\alpha$ from every alternative in $\{b_i \ldots, b_n, \ldots, b_{j-1},a_1, \ldots, a_n\}\setminus\{a_j\}$.
    
    \item[-] Agent $j > i$ is at distance $1$ from alternative $a_j$, 
    at distance $\alpha+\delta$ from every alternative in $\{b_j, \ldots, b_n\}$, 
    and at distance $2+\alpha$ from every alternative in $\{b_1, \ldots, b_{j-1},a_1, \ldots, a_n\}\setminus\{a_j\}$.
\end{itemize}
Clearly, the specified distances are consistent to the ordinal preferences of the agents (with ties broken as required). It is also not hard to verify that these distances, together with the fixed distances between alternatives, satisfy the triangle inequality; for an example of the metric space for the case $i=1$ see Figure~\ref{fig:thm:sc-general-all-three-types}. 
By interpreting $b_0$ as $b_n$, observe that every agent $j \neq i$ is at distance $\alpha+\delta$ from $b_{i-1}$, and distance $1+2\alpha$ from $b_i$ and $a_i$, and thus
$$\SC(a_i) = 1+(n-1)(2+\alpha),$$ 
$$\SC(b_i) = \alpha+\delta + (n-1)(2+\alpha),$$ 
and 
$$\SC(b_{i-1})=n(\alpha+\delta).$$ 
As $n$ becomes arbitrarily large and $\delta$ becomes arbitrarily small, we obtain a lower bound of $1+2/\alpha > 2$, since $\alpha < 2$, a contradiction.
\end{proof}

\subsection{Refined Lower Bounds for General Metric Spaces}

As we mentioned earlier, {\sc $\alpha$-Minisum-TAS-Distance} achieves a distortion of $1+\sqrt{2}$ without having access to the ordinal preferences of the agents. This begs the natural question of whether this is the best possible among all mechanisms with this information available in their inputs. We prove that this is indeed the case in the following theorem.

\begin{theorem} \label{thm:sc-general-lower-known}
For the social cost objective, the distortion of any mechanism $M \in \dis \cap \tas$ is at least $1+\sqrt{2}$.
\end{theorem}

\begin{proof}
Assume by contradiction that a mechanism with distortion strictly less than $1+\sqrt{2}$ exists, and consider any such mechanism. 
By \Cref{lem:sc-general-alpha-all-three-types}, any mechanism has distortion at least $\alpha$, even when the distances between alternatives are known. 
So, the mechanism must use $\alpha < 1+\sqrt{2}$. 

We consider instances with $n$ agents and $m=2n+2$ alternatives $\{x_1,...,x_n,z_1, ...,z_n,y,w\}$. The $\alpha$-TAS of agent $i$ is $A_i =\{x_i,z_i\}$. The metric space is presented in \Cref{fig:thm:list-lower}; in the figure, $\alpha^{\text{\footnotesize+}}$ is used to denote $\alpha$ plus some infinitesimal. 
Without loss of generality, we assume that the mechanism chooses one of $\{x_1,...,x_n,w\}$; the case where the mechanism chooses one of $\{z_1,...,z_n,y\}$ is symmetric. 

\begin{figure}[t]
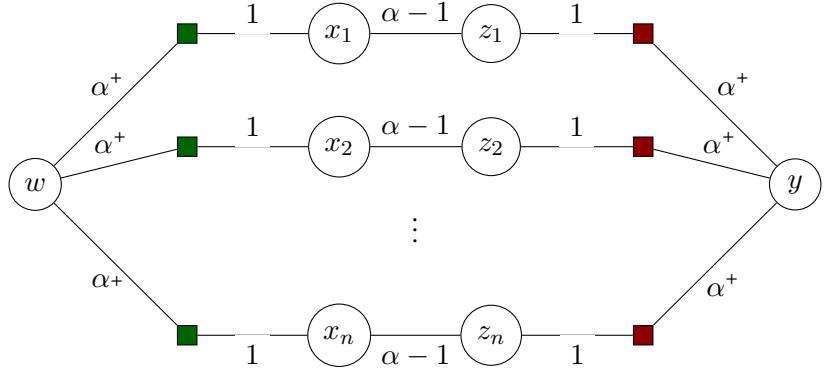

\centering

\tikz {
  \node (w) [circle, draw] at (0,0) {$w$};
  
  \node (g1) [rectangle, fill=DarkGreen, draw] at (2,2) {};
  \node (x1) [circle, draw] at (4,2) {$x_1$};
  \node (z1) [circle, draw] at (6,2) {$z_1$};
  \node (r1) [rectangle, fill=DarkRed, draw] at (8,2) {};

  \node (g2) [rectangle, fill=DarkGreen, draw] at (2,0.5) {};
  \node (x2) [circle, draw] at (4,0.5) {$x_2$};
  \node (z2) [circle, draw] at (6,0.5) {$z_2$};
  \node (r2) [rectangle, fill=DarkRed, draw] at (8,0.5) {};

  \node (gn) [rectangle, fill=DarkGreen, draw] at (2,-2) {};
  \node (xn) [circle, draw] at (4,-2) {$x_n$};
  \node (zn) [circle, draw] at (6,-2) {$z_n$}; 
  \node (rn) [rectangle, fill=DarkRed, draw] at (8,-2) {}; 

  \node (y) [circle, draw] at (10,0) {$y$};

  \node (dots) at (5,-0.5) {$\vdots$};
  
  \draw (w) edge node [above,xshift=-3pt] {$\alpha^{\text{\footnotesize+}}$} (g1) ;
  \draw (g1) edge node [above, fill=white] {$1$} (x1) ;
  \draw (x1) edge node [above, fill=white] {$\alpha-1$} (z1) ;
  \draw (z1) edge node [above, fill=white] {$1$} (r1) ;
  \draw (y) edge node [above, xshift=7pt] {$\alpha^{\text{\footnotesize+}}$} (r1) ;

\draw (w) edge node [above,xshift=-3pt] {$\alpha^{\text{\footnotesize+}}$} (g2) ;
  \draw (g2) edge node [above, fill=white] {$1$} (x2) ;
  \draw (x2) edge node [above, fill=white] {$\alpha-1$} (z2) ;
  \draw (z2) edge node [above, fill=white] {$1$} (r2) ;
  \draw (y) edge node [above, xshift=3pt] {$\alpha^{\text{\footnotesize+}}$} (r2) ;

  \draw (w) edge node [below,xshift=-3pt] {$\alpha{\text{\footnotesize+}}$} (gn) ;
  \draw (gn) edge node [below, fill=white] {$1$} (xn) ;
  \draw (xn) edge node [below, fill=white] {$\alpha-1$} (zn) ;
  \draw (zn) edge node [below, fill=white] {$1$} (rn) ;
  \draw (y) edge node [below,xshift=3pt] {$\alpha^{\text{\footnotesize+}}$} (rn) ;
}

\caption{The metric space used in the proof of Theorem~\ref{thm:sc-general-lower-known}. The notation $\alpha^{\text{\footnotesize+}}$ is used to represent a value slight larger than $\alpha$ (essentially $\alpha$ plus an infinitesimal). The red and green squares are potential agent positions, depending on the decision made by the mechanism when given that the $\alpha$-TAS of each agent $i$ is $A_i = \{x_i,z_i\}$.}
\label{fig:thm:list-lower}
\end{figure}

We place the agents at the red nodes in \Cref{fig:thm:list-lower}; for the case the case where the mechanism chooses the winner from the set $\{z_1,...,z_n,y\}$, we would place the agents at the green nodes in \Cref{fig:thm:list-lower}. Agent $i$ is at distance $1$ from $z_i$, at distance $1+\alpha-1 = \alpha$ from $x_i$, and at distance strictly larger than $\alpha$ from any other alternative; so, this placement of the agents is consistent with the $\alpha$-TAS. Since the distance of agent $i$ is essentially $\alpha$ from $y$, at least $1+2\alpha$ from $w$, and $3\alpha$ from any alternative $x_j$ with $j \neq i$, we have
$$\SC(x_i) = \alpha + 3(n-1)\alpha = (3n-2)\alpha,$$
$$\SC(w) = n(1+2\alpha),$$
and
$$\SC(y) = n\alpha.$$
So, the distortion is at least $3-2/n$ if any alternative $x_i$ is chosen, and at least $2 + 1/\alpha$ if $w$ is chosen. But, for $\alpha < 1+\sqrt{2}$, this leads to a lower bound of $2 + 1/\alpha > 1+\sqrt{2}$, a contradiction.
\end{proof}

We also show the same distortion lower bound of $1+\sqrt{2}$ for mechanisms that have access to the $\alpha$-TAS and the ordinal preferences of the agents, but not the distances between alternatives.

\begin{theorem} \label{thm:sc-general-lower-ordinal}
For the social cost objective, the distortion of any mechanism $M \in \ord \cap \tas$ is at least $1+\sqrt{2}$.
\end{theorem}

\begin{proof}
Assume by contradiction that a mechanism with distortion strictly less than $1+\sqrt{2}$ exists, and consider any such mechanism. 
By \Cref{lem:sc-general-alpha-all-three-types}, the distortion of any deterministic mechanism is at least $\min\{3,\alpha\}$, even when the ordinal preferences of the agents are given. So, the mechanism must use $\alpha < 1+\sqrt{2}$.

\begin{figure}[t]
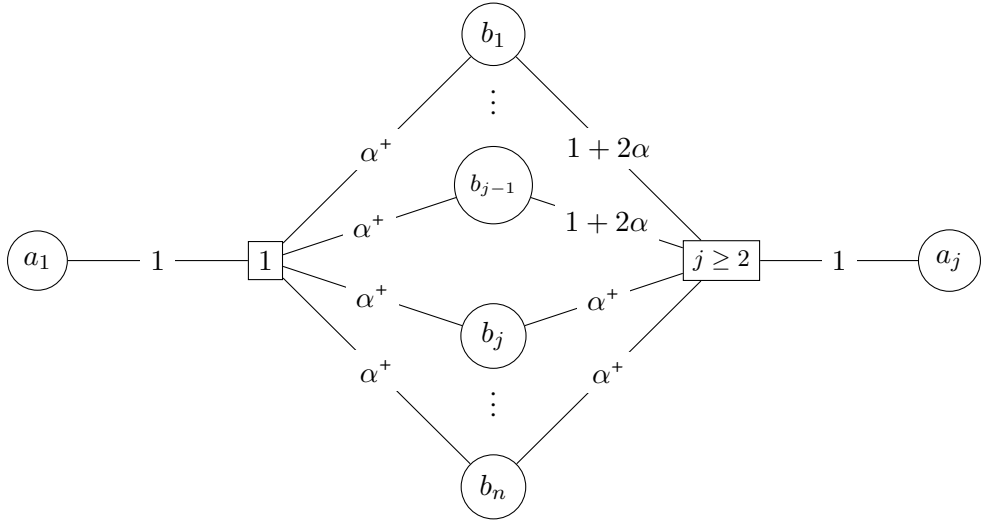

\centering
\tikz {
  \node (b1) [circle, draw] at (0,3) {$b_1$};
  \node (dots) at (0,2.2) {$\vdots$};
  \node (bi-1) [circle, draw] at (0,1)  {\footnotesize $b_{j-1}$};
  \node (bi) [circle, draw] at (0,-1)  {$b_j$};
  \node (dots) at (0,-1.8) {$\vdots$};
  \node (bn) [circle, draw] at (0,-3)  {$b_n$};
  
  \node (i) [rectangle, draw] at (-3,0) {$1$};
  \node (ai) [circle, draw] at (-6,0) {$a_1$};
  \node (j) [rectangle, draw]  at (3,0) {\footnotesize $j\geq 2$};
  \node (aj) [circle, draw] at (6,0) {$a_j$};

  \draw (i) edge node [midway, fill=white] {$1$} (ai) ;
  \draw (j) edge node [midway, fill=white] {$1$} (aj) ;

  \draw (i) edge node [midway, fill=white] {$\alpha^{\text{\footnotesize+}}$} (b1) ; 
  \draw (i) edge node [midway, fill=white] {$\alpha^{\text{\footnotesize+}}$} (bi-1) ; 
  \draw (i) edge node [midway, fill=white] {$\alpha^{\text{\footnotesize+}}$} (bi) ; 
  \draw (i) edge node [midway, fill=white] {$\alpha^{\text{\footnotesize+}}$} (bn) ; 
  
  \draw (j) edge node [midway, fill=white] {$1+2\alpha$} (b1) ; 
  \draw (j) edge node [midway, fill=white] {$1+2\alpha$} (bi-1) ; 
  \draw (j) edge node [midway, fill=white] {$\alpha^{\text{\footnotesize+}}$} (bi) ; 
  \draw (j) edge node [midway, fill=white] {$\alpha^{\text{\footnotesize+}}$} (bn) ; 
  
}
\caption{An example of the metric space used in the proof of \Cref{thm:sc-general-lower-ordinal} to show a lower bound of $2+1/\alpha$ for $i=1$. Agents are represented by rectangles, and alternatives by circles. The notation $\alpha^{\text{\footnotesize+}}$ is used to represent a value slightly larger than $\alpha$ (instead of $\alpha+\delta$).}
\label{fig:thm:sc-general-lower-ordinal}
\end{figure}

We consider an instance similar to the one used in the proof of Theorem~\ref{thm:sc-general-all-three-types} without the requirement to have fixed distances between alternatives (as this information is not used by the mechanism). There are $n$ agents and $2n$ alternatives $A=\{a_1,\ldots,a_n,b_1,\ldots,b_n\}$. 
Let $[X]$ be an arbitrary permutation of the elements of $X$. Each agent $i \in [n]$ has the ordinal preference
\begin{align*}
    a_i \succ_i b_i \succ_i \ldots \succ_i b_n \succ_i b_1 \succ \ldots \succ_i b_{i-1} \succ [\{a_1, \ldots, a_n\} \setminus \{a_i\}].
\end{align*}
In addition, agent $i$ reports that she only approves of alternative $a_i$, that is, $A_i = \{a_i\}$.
Let $\delta>0$ be an infinitesimal.
Now, suppose that, for some $i \in [n]$, the mechanism chooses $a_i$ or $b_i$ as the winner. 
We consider the following metric space:
\begin{itemize}
    \item[-] Agent $i$ is 
    at distance $1$ from alternative $a_i$, 
    at distance $\alpha+\delta$ from each alternative in $\{b_1, \ldots, b_n\}$, 
    and at distance $1+2(\alpha+\delta)$ from the remaining alternatives in $\{a_1, \ldots, a_n\}\setminus\{a_i\}$.

    \item[-] Agent $j < i$ is at distance $1$ from alternative $a_j$,    
    at distance $\alpha+\delta$ from every alternative in $\{b_j, \ldots, b_{i-1}\}$, 
    at distance $1+2\alpha$ from every alternative in $\{b_i, \ldots, b_n, \ldots, b_{j-1}\}$, 
    and at distance $1+2(\alpha+\delta)$ from every alternative in $\{a_1, \ldots, a_n\}\setminus\{a_j\}$.
    
    \item[-] Agent $j > i$ is at distance $1$ from alternative $a_j$, 
    at distance $\alpha+\delta$ from every alternative in $\{b_j, \ldots, b_n\}$, 
    at distance $1+2\alpha$ from every alternative in $\{b_1, \ldots, b_{j-1}\}$, 
    and at distance $1+2(\alpha+\delta)$ from every alternative in $a_1, \ldots, a_n\}\setminus\{a_j\}$.
\end{itemize}
An example of the metric space for the case $i=1$ is given in Figure~\ref{fig:thm:sc-general-lower-ordinal}. 
Interpreting $b_0$ as $b_n$, we can observe that every agent $j \neq i$ has value $\alpha+\delta$ for $b_{i-1}$, and value $1+2\alpha$ for $b_i$ and $a_i$, and thus
$$\SC(a_i) = 1+(n-1)(1+2\alpha+2\delta),$$ 
$$\SC(b_i) = \alpha+\delta + (n-1)(1+2\alpha),$$ 
and 
$$\SC(b_{i-1})=n(\alpha+\delta).$$ 
As $n$ becomes arbitrarily large and $\delta$ becomes arbitrarily small, we obtain a lower bound of $2+1/\alpha \geq 2+1/(1+\sqrt{2}) = 1+\sqrt{2}$, since $\alpha < 1+\sqrt{2}$.
\end{proof}

It is an interesting open problem whether there exists a mechanism that inputs only the $\alpha$-TAS and the ordinal preferences and achieves a distortion of $1+\sqrt{2}$ in general metric spaces. In the next subsection, we will show  that, for the line metric, there is a mechanism with these informational inputs that achieves the best possible distortion, even among mechanisms that use all three types of information.

Before we do that however, we show that any mechanism that uses only the $\alpha$-TAS (i.e., without having access to the distances between alternatives or the ordinal preferences) is bound to have a very high distortion.

\begin{theorem}\label{thm:sc-lower-only-approval-sets}
For the social cost objective, the distortion of any $M \in \tas$ is $\Theta(n)$.
\end{theorem}

\begin{proof}
Clearly, if we choose $\alpha=1$, the $\alpha$-TAS of each agent is a singleton that consists of her most-preferred alternative. To obtain an asymptotically linear upper bound, it suffices to consider the very simple mechanism that uses $\alpha=1$, and outputs the most-preferred alternative $o_i$ of an arbitrary agent $i$ as the winner $w$. 
Denoting by $o$ an optimal alternative, and using the triangle inequality, we have
\begin{align*}
    \SC(w) &= \sum_{j \in N} d(j,o_i) \\
    &\leq \sum_{j \in N} \bigg( d(j,o) + d(i,o) + d(i,o_i) \bigg) \\
    &\leq \sum_{j \in N} \bigg( d(j,o) + 2d(i,o) \bigg) \\
    &\leq (1+2n) \cdot \SC(o),
\end{align*}
and thus the distortion is $\Theta(n)$ when we only have access to the $\alpha$-TAS. 

For the lower bound, consider an arbitrary deterministic mechanism that has access only to the $\alpha$-TAS, and the following instance with $n$ agents and $n$ alternatives $\{a_1,\ldots,a_n\}$. 
Each agent $i$ approves only alternative $a_i$, that is, $A_i = \{a_i\}$. 
Since the given information is symmetric, the mechanism can choose any of the alternatives as the winner, say $a_1$. 
Let $\varepsilon > 0$ be an infinitesimal. We define the following metric space:
\begin{itemize}
    \item[-] Agent $1$ and alternative $a_1$ are both located at $1$ of the line of real numbers. 
    \item[-] Agent $i \in \{2,\ldots,n-1\}$ and alternative $a_i$ are located at a unique point within the $\varepsilon$-neighborhood of $0$, but not at $0$.
    \item[-] Agent $n$ and alternative $a_n$ are located at $0$.
\end{itemize}
Since every agent $i$ is located exactly where their most-preferred alternative $a_i$ is located, this metric space is consistent to the $\alpha$-TAS given as input, for any $\alpha$. As $\varepsilon$ tends to $0$, we have that $\SC(a_1) = n-1$, whereas $\SC(a_n) = 1$, leading to a lower bound of $n-1$.
\end{proof}


\subsection{Refined Results for the Line Metric}
When the metric space is a line, we will show refined tight bounds for mechanisms that use different types of information. The best possible distortion achievable here is $2\sqrt{2}-1 \approx 1.83$. The lower bound follows from the work of \citet{abramowitz2019awareness}, via a connection between their setting and ours, when there are only two alternatives; we provide a proof below using our terminology for completeness. More interestingly, by combining properties shown in previous work, we propose a mechanism, coined {\sc $\alpha$-Elimination-Weighted-Majority}, that achieves the distortion of $2\sqrt{2}-1 \approx 1.83$ as an upper bound for $\alpha = \sqrt{2}-1$, and is thus best possible. In fact, our mechanism does not need to have access to the distances between alternatives, only to the $\alpha$-TAS and the ordinal preferences of the agents. 

\begin{remark}\label{rem:ordinal-prefs-induce-positions}
Since the metric space is a line, we may assume that we have access to an ordering of the agents and alternatives on the line, which is unique up to permutations of identical agents and reversal; this ordering can be determined via the ordinal preferences of the agents \citep{elkind2014recognizing}.  
\end{remark}

Now, the mechanism works in two steps:
\begin{itemize}
\item {\em Elimination step:} 
It identifies the most-preferred alternative $x$ of the median agent, the alternative 
$\ell$ that is directly the left of $x$, and the alternative $r$ that is directly to the right of $x$. It then eliminates one of $\ell$ and $r$ by comparing the number $n(\ell,x)$ of agents that prefer $\ell$ to $x$ and the number $n(r,x)$ of agents that prefer $r$ to $x$ as in the work of \citet{anshelevich2017randomized}. In particular, if $n(\ell,x) \leq n(r,x)$, then it eliminates $\ell$, otherwise it eliminates $r$, and stores the non-eliminated alternative as $y$. 

\item 
{\em Weighted-majority step:} Afterwards, the mechanism runs a weighted majority between $x$ and $y$ by assigning a weight of $1$ to each agent $i$ such that $x, y \in A_i$, and a weight of $\frac{\alpha+1}{\alpha-1}$ to every other agent. This step will be shown to be equivalent to the algorithm used by \citet{abramowitz2019awareness} to show a bound of $2\sqrt{2}-1 \approx 1.83$ in their setting when there are two alternatives.
\end{itemize}
 See Mechanism~\ref{mech:elimination-weighted-majority} for a description using pseudocode. 
 
\newcommand\mycommfont[1]{\normalfont\textcolor{DarkGreen}{#1}}
\SetCommentSty{mycommfont}
\begin{algorithm}[t]
\SetNoFillComment
\caption{\sc $\alpha$-Elimination-Weighted-Majority}
\label{mech:elimination-weighted-majority}
{\bf Input:} Ordinal preferences, $\alpha$-TAS\;
{\bf Output:} Winner $w$\;
$x \gets $ most-preferred alternative of median agent\;
$\ell \gets $ alternative directly to the left of $x$\;
$r \gets $ alternative directly to the right of $x$\;
\tcp*[h]{Elimination step} \\
$n(\ell,x) \gets |\{i \in N: \ell \succ_i x \}|$\;
$n(r,x) \gets |\{i \in N: r \succ_i x \}|$\;
\uIf{$n(r,x) \geq n(\ell,x)$}{
    $y \gets r$\;
}
\Else{
    $y \gets \ell$\;
}
\For{$i \in N$}{
    \uIf{$x,y \in A_i$}{
        $w_i \gets 1$\;
    }
    \Else{
        $w_i \gets \frac{\alpha+1}{\alpha-1}$\;
    }
}
\tcp*[h]{Weighted-majority step} \\
$v_x \gets \sum_{i: x \succ_i y} w_i$\;
$v_y \gets \sum_{i: x \succ_i y} w_i$\;
\uIf{$v_x \geq v_y$}{
    $w \gets x$\;
}
\Else{
    $w \gets y$\;
}
\end{algorithm}

\begin{theorem}\label{thm:sc-line-ordinal-upper}
When the metric space is a line, for the social cost objective, the distortion of {\sc $\alpha$-Elimination-Weighted-Majority} is at most $\max\{\frac{3\alpha-1}{\alpha+1}, 1+\frac{2}{\alpha}\}$.
\end{theorem}

\begin{proof}
\citet{anshelevich2017randomized} first showed that one of the alternatives in $\{x,\ell,r\}$ must be the optimal alternative. In addition, they showed that $\ell$ cannot be better than $x$ in terms of social cost in case $n(\ell,x) \leq n(r,x)$, and $r$ cannot be better than $x$ in case if $n(\ell,x) > n(r,x)$. Hence, by appropriately eliminating one of $\ell$ or $r$ depending on whether $n(\ell,x) \leq n(r,x)$ or not, we have that the optimal alternative is one of $x$ or $y$ (the non-eliminated alternative among $\ell$ and $r$). Without loss of generality, we assume that $y=r$, and thus $y$ is at the right of $x$ on the line. We also assume that the mechanism chooses $x$ as the winner, in which case, to have a distortion larger than $1$, $y$ must be the optimal alternative; clearly, the case in which the mechanism chooses $y$ as the winner and $x$ is the optimal alternative can be handled with similar arguments. 

We partition the agents in the following sets:
\begin{itemize}
    \item $N_{xy} = \{i: x, y \in A_i \}$ is the set of agents that approve both $x$ and $y$;
    \item $N_{x\overline{y}} = \{i: x \in A_i, y \not\in A_i \}$ is the set of agents that approve only $x$;
    \item $N_{\overline{x}y} = \{i: x \not\in A_i , y \in A_i \}$ is the set of agents that approve only $y$;
    \item $N_{\overline{x} \succ \overline{y}} = \{i: x, y \not\in A_i, x \succ_i y\}$ is the set of agents that approve none of $x$ and $y$, but prefer $x$ over $y$.
    \item $N_{\overline{y} \succ \overline{x}} = \{i: x, y \not\in A_i, y \succ_i x \}$ is the set of agents that approve none of $x$ and $y$, but prefer $y$ over $x$.
\end{itemize}
Observe that the most-preferred alternative of any agent $i \in N_{\overline{x} \succ \overline{y}}$ must be at the left of $x$, and thus agent $i$ must also be at the left of $x$. Similarly, any agent $i \in N_{\overline{y} \succ \overline{x}}$ must be at the right of $y$. We argue that, in a worst-case instance, all agents lie in the interval $[x,y]$, and thus $N_{\overline{x} \succ \overline{y}} = N_{\overline{y} \succ \overline{x}} = \varnothing$. 
Let $i$ be an agent that is positioned at the left of $x$; clearly, agent $i \in N_{xy} \cup N_{x\overline{y}} \cup N_{\overline{x} \succ \overline{y}}$. 
The distortion can be written as 
\begin{align*}
\frac{\SC(x)}{\SC(y)}
= \frac
{\sum_{j \in N \setminus \{i\}} d(j,x) + d(i,x)}
{\sum_{j \in N \setminus \{i\}} d(j,y) + d(i,y)}.
\end{align*}
Since each agent $i$ is at the left of $x$, $d(i,y) = d(i,x) + d(x,y)$, and thus
\begin{align*}
\frac{\SC(x)}{\SC(y)}
&= \frac
{\sum_{j \in N \setminus \{i\}} d(j,x) + d(i,x)}
{\sum_{j \in N \setminus \{i\}} d(j,y) + d(i,x) + d(x,y)} 
\leq \frac
{\sum_{j \in N \setminus \{i\}} d(j,x)}
{\sum_{j \in N \setminus \{i\}} d(j,y) + d(x,y)}, 
\end{align*}
where the inequality follows since $d(i,x)$ appears in the numerator and the denominator of the ratio, and the distortion is, by definition, at least $1$. Hence, we can obtain a worse instance in terms of the distortion by moving any potential agent that are positioned at the left of $x$ to $x$; in this new instance, such agents can only approve $x$, and are thus all moved to the set $N_{x\overline{y}}$, leaving $N_{\overline{x} \succ \overline{y}}$ empty. A similar argument can be used to show that, in a worst-case instance, there are no agents at the right of $y$, and thus $N_{\overline{y} \succ \overline{x}} = \varnothing$.

Now, the distortion bound of $2\sqrt{2}-1$ follows by the corresponding bound shown by \citet{abramowitz2019awareness}. In their model, when there are two alternatives $a$ and $b$, an agent $i$ that prefers alternative $a$ over $b$ is said to have a preference strength of $\frac{d(i,b)}{d(i,a)}$. \citeauthor{abramowitz2019awareness} considered a weighted majority mechanism that assigns weight $1$ to all agents with preference strength at most $\tau$ and weight $\frac{\tau+1}{\tau-1}$ to all other agents (with preference strength strictly larger than $\tau$), for some $\tau \geq 1$, and showed that this mechanism achieves a distortion of at most $\max\{\frac{3\tau-1}{\tau+1}, 1+\frac{2}{\tau}\}$. 

In the worst-case instance of our mechanism, since all agents are positioned in the interval $[x,y]$, we have:
\begin{itemize}
    \item For any agent $i \in N_{xy}$, the preference strength of $i$ is at most $\alpha$. In particular, if $i$ prefers $x$ over $y$, $$d(i,y) \leq \alpha d(i,x) \Leftrightarrow \frac{d(i,y)}{d(i,x)} \leq \alpha.$$ Otherwise, $$d(i,x) \leq \alpha d(i,y) \Leftrightarrow \frac{d(i,x)}{d(i,y)} \leq \alpha.$$
    
    \item For any agent $i \in N_{x\overline{y}}$, the preference strength of $i$ is strictly larger than $\alpha$ since 
    $$d(i,y) > \alpha d(i,x) \Leftrightarrow \frac{d(i,y)}{d(i,x)} > \alpha.$$
    
    \item For any agent $i \in N_{\overline{x}y}$, the preference strength of $i$ is strictly larger than $\alpha$ since 
    $$d(i,x) > \alpha d(i,y) \Leftrightarrow \frac{d(i,x)}{d(i,y)} > \alpha.$$
\end{itemize}
Consequently, in the worst-case instance, the weighted majority that our mechanism runs in Step 3 (which assigns a weight of $1$ to each agent in $N_{xy}$ and a weight of $\frac{\alpha+1}{\alpha-1}$ to any other agent) is exactly the same as the weighted majority mechanism of \citet{abramowitz2019awareness} with $\tau = \alpha$, and thus the distortion is at most 
$\max\{\frac{3\alpha-1}{\alpha+1}, 1+\frac{2}{\alpha}\}$.
\end{proof}

By optimizing over alpha, we obtain the following corollary.  

\begin{corollary}
When the metric space is a line, for the social cost objective, the distortion of {\sc Elimination-Weighted-Majority} is at most $2\sqrt{2}-1 \approx 1.83$.
\end{corollary}

Next, we present a matching lower bound of $2\sqrt{2}-1$ on the distortion of any mechanism that uses all three types of information, which establishes the tightness of the distortion proven above. As mentioned before, this bound already follows by the work of \citet{abramowitz2019awareness}, whose model coincides with ours when there are only two alternatives.

\begin{theorem}\label{thm:sc-line-ordinal-lower}
When the metric space is a line, for the social cost objective, the distortion of any mechanism $M \in \ord \cap \dis \cap \tas$ is at least $2\sqrt{2}-1$.
\end{theorem}

\begin{proof}
Let $\varepsilon >0$ be an infinitesimal, and assume by contradiction that a mechanism with distortion strictly less than $2\sqrt{2}-1$ exists. Consider any such mechanism. We consider the following two instances with two alternatives $a$ and $b$, and two agents $1$ and $2$. 

\begin{itemize}
\item[-] First instance: $A_1=\{a\}$, $a \succ_1 b$, $A_2 = \{b\}$, and $b \succ_2 a$. Given this information, the mechanism can choose any of the two alternatives as the winner, say $a$. However, the positions of the agents might be the following ones: The two alternatives are at $0$ and $1$, and the two agents are at $\frac{1}{\alpha+1}-\varepsilon$ and $1$; it is not hard to verify that the distances are consistent with the information provided to the mechanism. Then, the social cost of $a$ is approximately $\frac{2+\alpha}{1+\alpha}$, whereas the social cost of $b$ is $\frac{\alpha}{1+\alpha}$, leading to a distortion of at least $1+\frac{2}{\alpha}$. Hence, to achieve a distortion strictly less than $2\sqrt{2}-1$, the mechanism must use $\alpha > 1+\sqrt{2}$.

\item[-] Second instance: $A_1=\{a,b\}$, $a \succ_1 b$, $A_2 = \{a,b\}$, and $b \succ_2 a$. Given this information, the mechanism can again choose any of the two alternative as the winner, say $a$. However, the positions of the agents might be the following ones: The two alternatives are at $0$ and $1$, and the two agents are at $1/2-\varepsilon$ and $\frac{\alpha}{1+\alpha}$; it is not hard to verify that the distances are consistent with the information provided to the mechanism. Then, the social cost of $a$ is approximately $\frac{3\alpha-1}{2(\alpha-1)}$, whereas the social cost of $b$ is $\frac{1+\alpha}{2(\alpha-1)}$, leading to a distortion of at least $\frac{3\alpha-1}{1+\alpha}$. For $\alpha > 1+\sqrt{2}$, $\frac{3\alpha-1}{1+\alpha} > 2\sqrt{2}-1$, a contradiction.
\end{itemize}
This completes the proof.
\end{proof}

We next consider the case of mechanisms that have access to the $\alpha$-TAS and the distances between alternatives, but not the ordinal preferences. 
We prove that the {\sc $2$-Minisum-TAS-Distance} mechanism (recall its definition in Mechanism~\ref{mech:minisum-tas-distance}) achieves a distortion of $2$ when the metric is a line, and this is the best possible among all mechanisms that use this type of information.

\begin{theorem}\label{thm:sc-line-locations-upper}
When the metric space is a line, the distortion of {\sc $\alpha$-Minisum-TAS-Distance} is at most $\max\{\alpha,1+\frac{2}{\alpha}\}$.
\end{theorem}

\begin{proof}
Since the metric space is a line, the sets $A_i$ define contiguous intervals. Let $\ell_i$ and $r_i$ denote the leftmost and rightmost alternatives in the $\alpha$-TAS $A_i$ of agent $i$, respectively. Let $o$ be an optimal alternative, and assume, without loss of generality, that $o$ is to the right of the alternative $w$ chosen by the mechanism. 

Consider the set $L = \{i: r_i \leq w\}$ of all agents for whom $A_i$ ends before $w$ or at $w$ itself, and the set $R = \{i: \ell_i > w\}$ of all agents whose intervals $A_i$ begin after $w$. Then, we claim that $|R| \leq |L|$. To see this, consider how the quantity $\sum_{i \in N} \min_{j \in A_i} d(j,w)$ compares to $\sum_{i \in N} \min_{j \in A_i} d(j,w')$, where $w'$ is the alternative directly to the right of $w$. When moving from $w$ to $w'$, this quantity decreases by $d(w,w')$ for each agent in $R$, increases by $d(w,w')$ for each agent in $L$, and stays the same for any other agent $i$ (since $w$ and $w'$ belong to $A_i$). Thus, by the definition of the mechanism, it must be that $|R|\leq |L|$, as otherwise $w'$ would be selected instead of $w$. 

We now proceed to prove the distortion bound. 
Since $|R|\leq |L|$, let $f:R\rightarrow L$ be a one-to-one mapping. 
We make the following observations:
\begin{itemize}
    \item[-] For every agent $i\in R$, let $j=f(i)\in L$. If $j$ is to the left of $w$, then 
    \begin{align*}
        d(j,w)+d(i,w) 
        &\leq d(j,w) + d(i,o) + d(w,o) \\
        &= d(j,o)+d(i,o).
    \end{align*}
    If, instead, $j$ is to the right of $w$, then it must be that $w$ is the closest alternative to $j$, by the definition of $L$. Therefore, since $o\not\in A_j$, we have that $d(j,w) \leq \frac{1}{\alpha}d(j,o)$. Combining this with the triangle inequality, we obtain
    \begin{align*}
    d(j,w)+d(i,w) 
    &\leq 2d(j,w)+d(j,o)+d(i,o) \\
    &\leq \left(1+\frac{2}{\alpha}\right) d(j,o) + d(i,o).
    \end{align*}
    \item[-] For every agent $i$ with $w\in A_i$, we have that 
    $$d(i,w) \leq \alpha d(i,o_i)\leq \alpha d(i,o).$$
    Notice that this must be true for every agent $i \not\in R\cup L$.
\end{itemize}
Combing the above, we obtain
\begin{align*}
\SC(w) 
    &=  \sum_{i \in L} d(i,w) + \sum_{i \in R} d(i,w) + \sum_{i \not\in R\cup L} d(i,w) \\
    &\leq  \left(1 + \frac{2}{\alpha} \right) \sum_{i \in L} d(i,o) + \sum_{i \in R} d(i,o) + \alpha \sum_{i \not\in R\cup L} d(i,o)\\
    &\leq \max\left\{\alpha, 1 + \frac{2}{\alpha}\right\} \cdot \SC(o). 
\end{align*}
Consequently, the distortion is at most $\max\left\{\alpha, 1 + \frac{2}{\alpha}\right\}$.
\end{proof}

By optimizing over $\alpha$, we obtain an upper bound of $2$.

\begin{corollary}
When the metric space is a line, the distortion of {\sc $2$-Minisum-TAS-Distance} is at most $2$.    
\end{corollary}

Below we prove the lower bound that establishes that {\sc $2$-Minisum-TAS-Distance} is best possible.

\begin{theorem}\label{thm:sc-line-locations-lower}
When the metric space is a line, for the social cost objective, the distortion of any mechanism $M \in \dis \cap \tas$ is at least $2$.
\end{theorem}

\begin{proof}
Assume by contradiction that a mechanism with distortion strictly less than $2$ exists, and consider any such mechanism. We consider the following two instances with alternative locations $0$ and $1+\alpha$. 

In the first instance, all agents include both alternatives in their $\alpha$-approval sets. As there is no way of distinguishing the two alternatives, the mechanism can choose any of them as the winner. If it chooses the one at $0$, then the agents may be located at $\alpha$, leading to a distortion of $\alpha$. 
If the mechanism chooses the alternative at $1+\alpha$, then the agents may be located at $1$, leading again to a distortion of $\alpha$. Overall, the distortion of the mechanism is at least $\alpha$. Therefore, the mechanism must require that $\alpha < 2$ to achieve distortion strictly less than $2$.

In the second instance, there are two agents with $A_1 = \{0\}$ and $A_2=\{1+\alpha\}$. Consider the following two cases depending on the decision made by the mechanism:
\begin{itemize}
    \item The winner is the alternative at $0$: 
    The first agent is located at $1-\varepsilon$ (so that the distance from $0$ is $1-\varepsilon$ while the distance from $1+\alpha$ is $\alpha+\varepsilon$, leading to the alternative at $0$ being the only one in $A_1$) and the second agent is located at $1+\alpha$. Hence, the social cost of $0$ is approximately $1 + 1+\alpha = 2+\alpha$, whereas the social cost of $1+\alpha$ is approximately $\alpha$, leading to a distortion of at least $1 + 2/\alpha$. 

    \item The winner is the alternative at $1+\alpha$: 
    The first agent is located at $0$, and the second agent is located at $\alpha+\varepsilon$ (so that the distance from $0$ is $\alpha+\varepsilon$ while the distance from $1+\alpha$ is $1-\varepsilon$, leading to the alternative at $1+\alpha$ being the one in $A_2$). Hence, the social cost of $0$ is approximately $\alpha$, whereas the social cost of $1+\alpha$ is approximately $1+\alpha + 1 = 2+\alpha$, leading to a distortion of at least $1+2/\alpha$. 
\end{itemize}
Overall, the distortion of the mechanism is least $1+2/\alpha$ in the second instance. Since $\alpha < 2$ from the first instance, we have that $1+2/\alpha > 2$, a contradiction.  
\end{proof}

Finally, we remark that the proof of  \Cref{thm:sc-lower-only-approval-sets} already uses an instance that is on the line metric, so we have the following straightforward corollary.

\begin{corollary}\label{cor:sc-lower-only-approval-sets}
For the line metric, for the social cost objective, the distortion of any mechanism $M \in \tas$ is $\Theta(n)$.
\end{corollary}


\section{Maximum Cost}


In this section we consider the other most natural objective, i.e., to minimize the maximum cost of any agent, abbreviated as $\text{MC}$. For mechanisms in $\ord \cap \dis \cap \tas$ we show a tight bound of $1+\sqrt{2}$. In the real line metric, we show that there in fact exist mechanisms in either $\ord \cap \tas$ and $\dis \cap \tas$ that achieve the $1+\sqrt{2}$ distortion guarantee, without having access to the third type of information. We also show that the best distortion achievable by any mechanism in $\tas$ is $3$. 

\subsection{Results for General Metric Spaces} \label{sec:mc-general}

We start by proving an upper bound of $1+\sqrt{2}$ for general metric spaces. This is achieved by a novel mechanism in $\ord \cap \dis \cap \tas$, which we refer to as {\sc $\alpha$-Most-Compact-Set}, for $\alpha =1+\sqrt{2}$. In fact, the mechanism requires less information than a typical mechanism in $\ord \cap \dis \cap \tas$, as it only requires access to the most-preferred alternatives of the agents, rather than their complete ordinal preferences. 

Before we define the mechanism, we will prove the following useful lemma relating the choices of any mechanism with the distortion for the maximum cost objective. To differentiate between alternatives that lie at minimum distance from an agent $i$ and alternatives that appear first in the preference ranking of the agent (which might be different due to ties), we will refer to the former as \emph{agent $i$'s min-distance alternative} (rather than agent $i$'s most-preferred alternative). In other words, a min-distance alternative for agent $i$ is an alternative $x$ for which $d(i,x) \leq d(i,x')$ for any $x' \in A$.

\begin{lemma}\label{lem:max-cost-upper-reference-lemma}
Consider a mechanism $M$ that chooses the winner to be an alternative $w$ that satisfies at least one of the following two conditions:
\begin{enumerate}
    \item  $w \in A_j$ for all $j \in N$.
    \item  $w$ is a min-distance alternative for some agent $i \in N$, and $o \notin A_i$, where $o$ is an alternative that minimizes the maximum cost.
\end{enumerate}
Then, for the maximum cost objective, the distortion of $M$ is at most $\max\{\alpha,2+\frac{1}{\alpha}\}$.
\end{lemma}

\begin{proof}
For the first condition, if $w \in A_j$ for all $j \in N$, then $d(j,w) \leq \alpha \cdot d(j,o_j) \leq \alpha \cdot d(j,o)$, where $o_j$ is the most-preferred alternative of agent $j$. This implies that $\MC(w) \leq \alpha \cdot \MC(o)$.
For the second condition, let $t$ be an agent that is at maximum distance from $w$. We have that
\begin{align*}
\MC(w) = d(t,w) &\leq d(t,o) + d(o,i) + d(i,w) \\
&< d(t,o) + d(i,o) + \frac{1}{\alpha}\cdot d(i,o) \\
&\leq \left(2+\frac{1}{\alpha}\right)\MC(o),
\end{align*}
where the first inequality is due to the triangle inequality, and the second inequality follows from the fact that $o \notin A_i$. 
\end{proof}

Now, our mechanism, coined {\sc $\alpha$-Most-Compact-Set}, works as follows: 
It first checks if there exists some alternative $x$ that is approved by all agents, and, if this is indeed the case, it then chooses $x$ as the winner $w$. 
Otherwise, for each agent $i$, it computes the maximum distance $\rho_i$ of the most-preferred alternative $o_i$ of $i$ to any other alternative in $A_i$; then it chooses the most-preferred alternative $o_i$ of the agent $i$ with minimum $\rho_i$. Essentially, in the otherwise case, the mechanism aims to choose the most-preferred alternative of an agent in order to minimize the radius to all other alternatives approved by the agent. See Mechanism~\ref{mech:most-compact-set} for a description using pseudocode. 

\begin{algorithm}[t]
\SetNoFillComment
\caption{\sc $\alpha$-Most-Compact-Set}
\label{mech:most-compact-set}
{\bf Input:} Most-preferred alternatives of agents, distances between alternatives, $\alpha$-TAS\;
{\bf Output:} Winner $w$\;
\uIf{$\forall i \in N, \exists x \in A_i$ }{
    $w \gets x$\;
}
\Else{
    \For{$i \in N$}{
	$\rho_i \gets \max_{x \in A_i} d(x,o_i)$\;
    }
    $w \in \arg\min_{i \in N} \rho_i$ \; 
}
\end{algorithm}

\begin{theorem}\label{thm:mc-general-upper}
In general metric spaces, for the max cost objective, the distortion of {\sc $\alpha$-Most-Compact-Set} is at most $\max\{\alpha, 2+\frac{1}{\alpha}\}$.
\end{theorem}

\begin{proof}
If there is an alternative $x$ such that $w \in A_j$ for all $j \in N$, the mechanism chooses it as the winner $w$. In that case, by \Cref{lem:max-cost-upper-reference-lemma}, the distortion of {\sc $\alpha$-Most-Compact-Set} is at most $\max\{\alpha, 2+\frac{1}{\alpha}\}$. If such an alternative does not exist, then the mechanism chooses the most-preferred alternative of some agent $i$, in particular the agent with the smallest $\rho_i$, i.e., the smallest maximum distance from any of her approved alternatives to her most-preferred alternative $o_i$. If $o \notin A_i$, then by \Cref{lem:max-cost-upper-reference-lemma}, the distortion of the mechanism is again at most $\max\{\alpha, 2+\frac{1}{\alpha}\}$. Therefore we will consider the case when $o \in A_i$. By definition of $\rho_i$, this implies that $d(o,o_i) \leq \rho_i$. 

Since there is no alternative $x$ such that $x \in A_i$ for all $i \in N$, there must exist at least one agent $j$ such that $o\not\in A_j$. Also, let $t$ be the agent that has the maximum distance from $w$. 
We have that
\begin{align*}
\MC(w) = d(t,o_i) 
&\leq d(t,o) + d(o,o_i) 
 \leq \MC(o) + \rho_i 
 \leq \MC(o) + \rho_j.
\end{align*}
where the last inequality follows from that fact that $i$ was chosen to be an agent that minimizes $\rho_k=\max_{x \in A_k}d(x,o_k)$ by the mechanism. 
Let $y$ be an alternative in $A_j$ with maximum distance from $o_j$. Then,
\begin{align*}
\rho_j 
= d(y,o_j) 
\leq d(j,o_j) + d(j,y) 
\leq \left(1 + \frac{1}{\alpha} \right) \MC(o),
\end{align*}
where the first inequality is due to the triangle inequality, and the second one is follows since $d(j,o_j) \leq d(j,o) \leq \MC(o)$, and since $o\not\in A_j$ means that $d(j,o_j) < \frac{1}{\alpha} d(j,o) \leq \frac{1}{\alpha} \MC(o)$. 
Putting everything together, we obtain 
\begin{align*}
\MC(w) \leq \left(2+\frac{1}{\alpha}\right) \MC(o).
\end{align*}
Overall, the distortion of the mechanism is at most $\max\{\alpha, 2+\frac{1}{\alpha}\}$, as desired.
\end{proof}

By optimizing over $\alpha$, we obtain the following corollary. 

\begin{corollary}
In general metric spaces, for the max cost objective, the distortion of {\sc $(1+\sqrt{2})$-Most-Compact-Set} is at most $1+\sqrt{2}$.
\end{corollary}

\noindent Next, we prove that the bound of $1+\sqrt{2}$ is the best possible for the maximum cost objective, even when all three types of information are available, and even when the metric space is as simple as a line.  

\begin{theorem}\label{thm:mc-lower}
For the max cost, the distortion of any mechanism $M \in \ord\cap\dis\cap\tas$ is at least $1+\sqrt{2}$, even on a line metric. 
\end{theorem}

\begin{proof}
Consider any mechanism $M$ that has access to all three types of information. 
To show the lower bound, we will consider the following two cases depending on the value of $\alpha$ that $M$ uses. 

\medskip
\noindent 
\textbf{Case 1: $\alpha \leq 1 + \sqrt{2}$.}
Consider an instance $I_1$ with two agents $\{1,2\}$ and two alternatives $\{a_1,a_2\}$. We have $A_1=\{a_1\}$, $A_2=\{a_2\}$, $a_1 \succ_1 a_2$, and $a_2 \succ_2 a_1$. The distance between the two alternatives is equal to $1+\alpha+\varepsilon$, for some infinitesimal $\varepsilon > 0$. Given this information, $M$ may choose any of the two alternatives as the winner, say $a_1$. We define the following line metric (see~\cref{fig:mc-lower-a}):
\begin{itemize}
    \item[-] Alternative $a_1$ is at $0$;
    \item[-] Agent $1$ is at $1$;
    \item[-] Alternative $a_2$ is at $1+\alpha + \varepsilon$;
    \item[-] Agent $2$ is at $1+2\alpha+\varepsilon$.
\end{itemize}
It is not hard to verify that these positions of the agents and the alternatives are consistent to the information given to the mechanism for any $\alpha \leq 1 + \sqrt{2}$: 
Agent $2$ approves only alternative $a_2$ as $d(2,a_1) = 1+2\alpha + \varepsilon > \alpha^2 = \alpha \cdot d(2,a_2)$; all other information is clearly consistent. 
We have that $\MC(a_1)=1+2\alpha+\varepsilon$ and $\MC(a_2)=\alpha+\varepsilon$, thus obtaining a lower bound of $2+1/\alpha \geq 1+\sqrt{2}$ when $\varepsilon$ tends to $0$.

\begin{figure}[t]
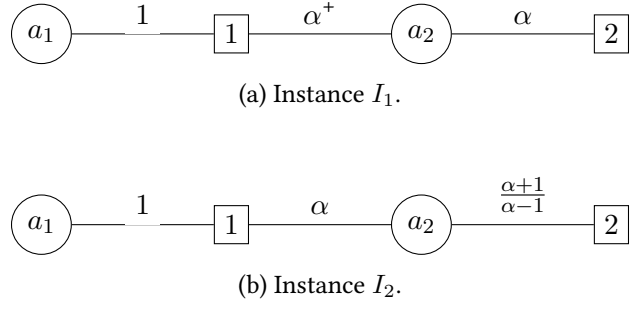

\centering
\begin{subfigure}[t]{\linewidth}
\centering
\tikz {
  \node (a1) [circle, draw] at (0,0) {$a_1$};
  \node (1) [rectangle, draw] at (2.5,0) {$1$};
  \node (a2) [circle, draw] at (5,0) {$a_2$};  
  \node (2) [rectangle, draw] at (7.5,0) {$2$};
  
  \draw (a1) edge node [above, fill=white] {$1$} (1) ; 
  \draw (1) edge node [above, fill=white] {$\alpha^{\text{\footnotesize+}}$} (a2) ; 
  \draw (a2) edge node [above, fill=white] {$\alpha$} (2) ;   
}
\caption{Instance $I_1$.}
\label{fig:mc-lower-a}
\end{subfigure}
\\[20pt]
\begin{subfigure}[t]{\linewidth}
\centering
\tikz {
  \node (a1) [circle, draw] at (0,0) {$a_1$};
  \node (1) [rectangle, draw] at (2.5,0) {$1$};
  \node (a2) [circle, draw] at (5,0) {$a_2$};  
  \node (2) [rectangle, draw] at (7.5,0) {$2$};
  
  \draw (a1) edge node [above, fill=white] {$1$} (1) ; 
  \draw (1) edge node [above, fill=white] {$\alpha$} (a2) ; 
  \draw (a2) edge node [above, fill=white] {$\frac{\alpha+1}{\alpha-1}$} (2) ;   
}
\caption{Instance $I_2$.}
\label{fig:mc-lower-b}
\end{subfigure}
\caption{The line metrics used in the proof of \Cref{thm:mc-lower}. Agents are represented by rectangles and alternatives are represented by circles. A value above an edge represents the distance between the edge's endpoints.}
\label{fig:mc-lower}
\end{figure}

\medskip
\noindent 
\textbf{Case 2: $\alpha > 1 + \sqrt{2}$.} 
Consider an instance $I_2$ with two agents $\{1,2\}$ and two alternatives $\{a_1,a_2\}$. We have $A_1=A_2=\{a_1,a_2\}$, $a_1 \succ_1 a_2$ and $a_2 \succ_2 a_1$. The distance between the two alternatives is equal to $1+\alpha$. Given this information, $M$ may choose any of the two alternatives as the winner, say $a_1$. We define the following line metric (see~\cref{fig:mc-lower-b}):
\begin{itemize}
    \item[-] Alternative $a_1$ is at $0$;
    \item[-] Agent $1$ is at $1$;
    \item[-] Alternative $a_2$ is at $\alpha+1$; 
    \item[-] Agent $2$ is at $\alpha \frac{\alpha+1}{\alpha-1}$.
\end{itemize}
It is not hard to verify that these positions of the agents and the alternatives are consistent to the information given to the mechanism for any $\alpha > 1 + \sqrt{2}$: 
Agent $1$ is at distance $1$ from $a_1$ and distance $\alpha$ from $a_2$. 
Agent $2$ is at distance $\alpha \frac{\alpha+1}{\alpha-1} - (\alpha+1) = \frac{\alpha+1}{\alpha-1}$ from $a_2$ and distance $\alpha+1$ from $a_1$; observe that the ratio of these distances is exactly $\alpha$. 
We have that $\MC(a_1) = \alpha \frac{\alpha+1}{\alpha-1}$ and $\MC(a_2) = \frac{\alpha+1}{\alpha-1}$ (since  $\frac{\alpha+1}{\alpha-1} \geq \alpha$ for any $\alpha > 1+\sqrt{2}$), thus obtaining a lower bound of $\alpha > 1+\sqrt{2}$. 
\end{proof}

We conclude the exposition of our results for general metric spaces by showing that the very simple mechanism {\sc Any-Approved} $\in \tas$ that chooses any alternative in the $\alpha$-TAS of some agent achieves a distortion bound of $2+\alpha$ in terms of the maximum cost objective. By setting $\alpha = 1$, we obtain a distortion bound of $3$, which is the best possible for all mechanisms in $\tas$, and also matches the bound of $3$ that is best possible for any mechanism in $\ord\cap \dis$. 

\begin{theorem}\label{MC-upper-only-sets}
In general metric spaces, for the max cost objective, the distortion of {\sc Any-approved} is at most $2+\alpha$.
\end{theorem}

\begin{proof}
Consider an arbitrary instance and let $w$ be the alternative chosen by $M$ on this instance. By definition, $w$ is in the $\alpha$-TAS of some $i$. 
Let $j$ be an agent that has the maximum distance from $w$. 
Let $o$ be an optimal alternative, i.e., the one that minimizes the maximum cost. 
Let $o_i$ be the alternative that is closest to the position of agent $i$. 
We have that
\begin{align*}
\MC(w) = d(j,w) &\leq d(j,o) + d(i,o) + d(i,w) \\
& \leq 2\MC(o) + \alpha \cdot d(i, o_i)\\
& \leq (2+\alpha) \MC(o),
\end{align*}
where the first inequality holds by the triangle inequality, the second inequality follows from the definition of $o$ and the fact that $w \in A_i$, 
and the last inequality follows since $d(i,o_i) \leq d(i,o)$ by the definition of $o_i$.
\end{proof}

The tight lower bound for all mechanisms in $\tas$ is shown below. 

\begin{theorem}\label{thm:max-tas-lower}
For the max cost objective, the distortion of any mechanism $M \in \tas$ is at least $3$, even when the metric space is a line.
\end{theorem}

\begin{proof}
Consider an arbitrary mechanism $M \in \tas$, and an instance with two agents $\{1,2\}$, and five alternatives $\{a_1,\ldots,a_5\}$. That $\alpha$-TAS of the agents are $A_1=\{a_1,a_2\}$ and $A_2=\{a_3,a_4\}$. Let $\varepsilon$ be an infinitesimal. 
If $M$ chooses alternative $a_5$ as the winner, then the distortion can easily be seen to be unbounded as the metric space might be the following: 
\begin{itemize}
    \item[-] Agent $1$ and alternatives $a_1, a_2$ are at $0$; 
    \item[-] Agent $2$ and alternatives $a_3,a_4$ are at $\varepsilon$;
    \item[-] Alternative $a_5$ is at $1$.
\end{itemize}
Hence, the optimal max cost is $\varepsilon$, whereas the max cost of the mechanism is $1$.

\begin{figure}[t]
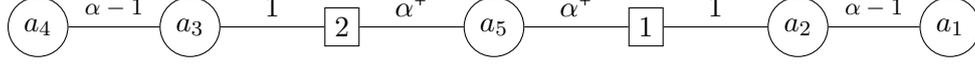

\centering
\tikz {
  \node (a4) [circle, draw] at (-6,0) {$a_4$};
  \node (a3) [circle, draw] at (-4,0) {$a_3$};  
  \node (2) [rectangle, draw] at (-2,0) {$2$};
  \node (a5) [circle, draw] at (0,0) {$a_5$};  
  \node (1) [rectangle, draw] at (2,0) {$1$};
  \node (a2) [circle, draw] at (4,0) {$a_2$};  
  \node (a1) [circle, draw] at (6,0) {$a_1$};
  
  \draw (a4) edge node [above, fill=white] {\footnotesize$\alpha-1$} (a3) ; 
  \draw (a3) edge node [above, fill=white] {$1$} (2) ; 
  \draw (2) edge node [above, fill=white] {$\alpha^{\text{\footnotesize+}}$} (a5) ; 

  \draw (a1) edge node [above, fill=white] {\footnotesize$\alpha-1$} (a2) ; 
  \draw (a2) edge node [above, fill=white] {$1$} (1) ; 
  \draw (1) edge node [above, fill=white] {$\alpha^{\text{\footnotesize+}}$} (a5) ;   
}
\caption{The line metric used in the proof of \Cref{thm:max-tas-lower}.}
\label{fig:max-tas-lower}
\end{figure}

So, suppose that $M$ selects one of the alternatives that is approved by one of the agents, say $w=a_1$. 
We define the following line metric (see also Figure~\ref{fig:max-tas-lower}):
\begin{itemize}
    \item[-] Agent 1 is at $\alpha+\varepsilon$;
    \item[-] Alternative $a_2$ is at $\alpha + 1 + \varepsilon$;
    \item[-] Alternative $a_1$ is at $2\alpha+\varepsilon$
    \item[-] Agent 2 is at $-\alpha-\varepsilon$;
    \item[-] Alternative $a_3$ is at $-\alpha - 1 - \varepsilon$;
    \item[-] Alternative $a_4$ is at $-2\alpha-\varepsilon$;
    \item[-] Alternative $a_5$ is at $0$.
\end{itemize}
It is not hard to verify that the positions of the agents and the alternatives on the line are consistent to the $\alpha$-TAS: 
Agent $1$ is at distance $1$ from $a_2$, $\alpha$ from $a_1$, and strictly larger than $\alpha$ from any other alternative. Similarly, agent $2$ is at distance $1$ from $a_3$, $\alpha$ from $a_4$, and strictly larger than $\alpha$ from any other alternative. The cost of the mechanism is $\MC(w) = d(2,a_1) \approx 3\alpha$, whereas the optimal cost if $\MC(a_5) = d(1,a_5) = d(2,a_5) \approx \alpha$, leading to a distortion of at least $3$.  
\end{proof}


\subsection{Refined Upper Bounds for the Line Metric} \label{sec:mc-line}

The lower bounds of $1+\sqrt{2}$ for mechanisms in $\ord\cap\dis\cap\tas$ (\Cref{thm:mc-lower}) and $3$ for mechanisms in $\tas$ (\Cref{thm:max-tas-lower}) already apply when the metric space is a line. For both cases, tight upper bounds can be achieved for general spaces (\Cref{thm:mc-general-upper} and \Cref{MC-upper-only-sets}). In this section, we show that, when the metric space is a line, the bound of $1+\sqrt{2}$ can be achieved by several mechanisms in $\ord \cap \tas$ and $\dis \cap \tas$ (i.e., without having access to the third type of information). Whether such mechanisms exist for general metric spaces  is an intriguing open problem. 

From a technical perspective, these mechanisms, although different in nature, establish the same kind of properties that are needed for the distortion guarantees to hold. For this reason, we first present all these mechanisms, and then prove their distortion bounds with one, unified proof. Before that, we start with a simple known observation for the maximum cost objective on the line.

\begin{observation}\label{obs:within-agent-interval}
Any mechanism that always chooses the winner to be an alternative that lies in the interval defined by the leftmost and rightmost agents has distortion at most $2$.
\end{observation}

\begin{proof}
Let $\ell$ and $r$ be the positions of the leftmost and rightmost agent, respectively. Clearly, the max cost of the mechanism is $\MC(w) \leq d(\ell,r)$, whereas the optimal max cost is $\MC(o) \geq \frac{1}{2}\cdot d(\ell,r)$, and thus the distortion is at most $2$.
\end{proof}

We now present the mechanisms, starting with a member of the class $\ord \cap \tas$, to which we refer as \textsc{Max-TAS-Leftmost}. 
The mechanism computes the set $S^{\tas}$ of alternatives that are members of the most $\alpha$-TAS, and chooses the winner $w$ to be the leftmost alternative in this set; note that is latter is possible on the line given the ordinal preferences, see \Cref{rem:ordinal-prefs-induce-positions}. See Mechanism~\ref{mech:max-tas-min-rank} for a description using pseudocode. The next two mechanisms we consider are members of $\dis \cap \tas$. The first of them is the \textsc{$\alpha$-Minisum-TAS-Distance} that we defined and used in \Cref{sec:sc-general}, which chooses the winner $w$ to minimize the sum of distances from all $\alpha$-TAS; recall its definition in Mechanism~\ref{mech:minisum-tas-distance}.  The second one is a mechanism to which we refer as \textsc{$\alpha$-Minimax-TAS-Distance}. This mechanism is the natural translation of \textsc{$\alpha$-Minisum-TAS-Distance} to the maximum cost objective, as it chooses $w$ to minimize the maximum distance from any $\alpha$-TAS. See Mechanism~\ref{mech:minimax-tas-distance} for its description.

\subsubsection*{Distortion Upper Bounds}

We now show a proof that applies to all the mechanisms defined above at the same time, by establishing a series of claims. For brevity, we will refer to the mechanisms as $M_1$ for {\sc Max-TAS-Leftmost}, $M_2$ for {\sc $\alpha$-Minisum-$\tas$-Distance}, and $M_3$ for {\sc $\alpha$-Minimax-$\tas$-Distance}.

\begin{theorem}\label{thm:mc-line-upper}
When the metric space is a line, for the maximum cost objective, the distortion of any mechanism $M \in \{M_1,M_2,M_3\}$ is at most $\max\{\alpha,2+\frac{1}{\alpha}\}$. 
\end{theorem}

\begin{proof}
Assume by contradiction that the distortion of $M$ is larger than $\max\{\alpha,2+\frac{1}{\alpha}\}$. Consider an arbitrary instance, and let $i_1$ and $i_n$ be the positions of the leftmost and rightmost agent on the line respectively (breaking ties arbitrarily). Let $w$ be the alternative chosen by $M$ for the given instance, and let $o$ be an optimal alternative that minimizes the maximum cost. 
Clearly, $w \neq o$, as, otherwise, the distortion of $M$ would be $1$. 
Additionally, $w \notin [i_1,i_n]$, as, otherwise, by \Cref{obs:within-agent-interval} the distortion of $M$ would be at most $2$. 
Without loss of generality, we assume that $w > i_n$, and thus $o < w$. 
Finally, there exists some agent $j \in N$ such that $w \not\in A_j$, as, otherwise, the distortion of $M$ would be most $\max\{\alpha,2+\frac{1}{\alpha}\}$ by \Cref{lem:max-cost-upper-reference-lemma}. 

For any alternative $x$, let $\ell(x)$ denote the alternative directly to the left of $x$ on the line, which does not coincide with $x$ (breaking ties arbitrarily). 
We state and prove the following useful claim.

\begin{algorithm}[t]
\SetNoFillComment
\caption{\sc Max-TAS-Leftmost}
\label{mech:max-tas-min-rank}
{\bf Input:} Ordinal preferences, $\alpha$-TAS\;
{\bf Output:} Winner $w$\;

\For{$x \in A$}{
    $n_x \gets |\{i \in N: x \in A_i\}|$\;
}
$S^{\tas} \gets \arg\max_{x \in A}  n_x$\;
$w \gets \text{leftmost alternative in } S^\tas$;
\end{algorithm}

\begin{algorithm}[t]
\SetNoFillComment
\caption{\sc $\alpha$-Minimax-TAS-Distance}
\label{mech:minimax-tas-distance}
{\bf Input:} Distances between alternatives, $\alpha$-TAS\;
{\bf Output:} Winner $w$\;
$w \in \arg\min_{x \in A} \bigg\{ \max_{i \in N} \min_{j \in A_i} d(j,x) \bigg\}$\;
\end{algorithm}

\begin{claim}\label{claim:best-for-in}
$w$ is a min-distance alternative for agent $i_n$.
\end{claim}

\begin{proof}
Assume by contradiction that $w$ is not a min-distance alternative for $i_n$. Then $\ell(w)$ has to be min-distance alternative for agent $i_n$; note that such an alternative exists since $o < w$. Then, for every agent $i \leq i_n$ we have that $w \in A_i \Rightarrow \ell(w) \in A_i$, since $w > \ell(w)$. 
This is enough to establish the claim for the case when $M=M_1$. Indeed, $\ell(w)$ appears in at least as many $\alpha$-TAS as $w$, 
and $\ell(w) < w$, contradicting the choice of $w$ by the mechanism.

Next we prove the claim for for $M \in \{M_2,M_3\}$. As we established in the first paragraph, all agents that find $w$ acceptable also find $\ell(w)$ acceptable and thus for those we have $d(\ell(w),A_j) = d(w,A_j)= 0$. Now consider any agent $j \in N$ such that $w \notin A_j$ and let $S^{\neg w}$ be the set of those agents. Since $w$ is not acceptable for all agents, we have that $|S^{\neg w}| \geq 1$. For any agent $j \in S^{\neg w}$, let $x_r^{j}$ be the rightmost alternative smaller than $w$ in $A_j$ (breaking ties arbitrarily), and observe that $x_r^j \leq \ell(w)$. We have that $d(x_r^j,\ell(w)) < d(x_r^j,w)$. Since this holds for every agent $j \in S^{\neg w}$, this means that $\ell(w)$ has both a smaller total and maximum distance from the $\alpha$-TAS, contradicting the choice of $w$ by the mechanism. 
\end{proof}

Given \Cref{claim:best-for-in}, by \Cref{lem:max-cost-upper-reference-lemma}, it suffices to prove that $o \notin A_{i_n}$, as, otherwise, the distortion would be upper bounded by $\max\{2+\frac{1}{\alpha}\}$, leading to a contradiction. This is established by the following claim, the proof of which is similar to that of \Cref{claim:best-for-in}.

\begin{claim}\label{claim:o-notin-Ai}
$o \notin A_{i_n}$.
\end{claim}

\begin{proof}
Assume by contradiction that $o \in A_{i_n}$. Let $S^w$ be the set of agents for which $w$ is a min-cost alternative, and notice that $i_n \in S^w$ be \Cref{claim:best-for-in}. Notice that for every agent $i \in S^w$, and for any alternative $x \in [o, w]$ we have that $x \in A_i$; in particular, $\ell(w) \in A_i$ for those agents. For agents $i \notin S^w$, it holds that $d(i,\ell(w)) < d(i,w)$, as (a) either $i \geq \ell(w)$ and $\ell(w)$ is a min-cost alternative for $i$, or (b) $i < \ell(w)$. Overall, we have that for any agent $j \in N$, $w \in A_j \Rightarrow \ell(w) \in A_j$, similarly to \Cref{claim:best-for-in}. Again, for the case of $M=M_1$, since $\ell(w) < w$, this contradicts the choice of $w$ as the winner by the mechanism. The proof for $M \in \{M_2,M_3\}$ is identical to that of \Cref{claim:best-for-in}.
\end{proof}
This completes the proof. 
\end{proof}

By optimizing over $\alpha$, we obtain the following corollary, establishing that all these mechanisms are best possible when the metric space is a line. 

\begin{corollary}\label{cor:mc-line-upper}
When the metric space is a line, for the maximum cost objective, the distortion of any mechanism $M \in \{M_1,M_2,M_3\}$ using $\alpha=1+\sqrt{2}$ is at most $1+\sqrt{2}$. 
\end{corollary}

\begin{remark}
One can obtain the distortion bound of \Cref{thm:mc-line-upper} using a mechanism in $\dis \cap \tas$ by appropriately modifying a mechanism of \citet{anshelevich2022distortion}. \citeauthor{anshelevich2022distortion} considered a different setting, that of \emph{distributed metric social choice}, where the agents are partitioned into districts, and local decisions within districts are then aggregated at a global level. In their setting, they define a mechanism called $\alpha$-\textsc{Acceptable-Rightmost-Leftmost} ($\alpha$-ARL), which first selects, for each district, the rightmost alternative that the agents therein find ``$\alpha$-acceptable'' and then, the leftmost among these alternatives is chosen as the winner. If we consider districts of size $1$ in their setting, $\alpha$-acceptability reduces to acceptability according to the $\alpha$-TAS in our setting. Note that their mechanism requires the ordering of the alternatives on the line to be able to identify the rightmost and leftmost of them. We can infer this information from the distances between alternatives on the line, and then appropriately modify their mechanism. Theorem~4.8 from their work would then yield the upper bound of our \Cref{thm:mc-line-upper}. Still, we believe that the mechanisms that we consider in this section are much more natural for our setting and much easier to describe.
\hfill $\qed$
\end{remark}


\section{Future Directions} \label{sec:open}
The main problem that our work leaves open is that of identifying the best possible distortion bound in terms of the social cost for general metric spaces and mechanisms in $\ord \cap \dis \cap \tas$, for which we showed an upper bound of $1+\sqrt{2}$ and a lower bound of $2$. Another interesting open question is to show whether $1+\sqrt{2}$ can be achieved with a mechanism in $\ord \cap \tas$ for general metric spaces. For the maximum cost objective, we were able to obtain tight bounds for the general class $\ord \cap \dis \cap \tas$, so the open question is whether these bounds can be achieved by mechanisms in $\ord \cap \tas$ and $\dis \cap \tas$, as we proved to be the case on the real line. For $\dis \cap \tas$, the $(1+\sqrt{2})$-\textsc{Minimax-TAS-Distance} mechanism, which we showed to achieve the tight bound of $1+\sqrt{2}$ on the line, seems like a very natural candidate. Unfortunately, however, we can construct counterexamples where the distortion of this mechanism is lower-bounded by $3$ in general metric spaces. For $\ord \cap \tas$, \textsc{Max-TAS-Leftmost} is obviously not even well-defined in general metric spaces, as the notion of a leftmost alternative is meaningless. That being said, we can design mechanisms in $\ord \cap \tas$ that are based on very similar principles and are not line-specific, which still achieve a distortion of $1+\sqrt{2}$ on the line. Unfortunately, once considered in general metric spaces, these mechanisms also fall short, as similar examples to the aforementioned one show a lower bound of $3$ on their distortion. 

More generally, our paper advocates the study of the interplay between information and efficiency in the metric social choice setting, building on a growing literature that has recently gained significant momentum. As a next step, one could consider eliciting additional cardinal information, for example, by using multiple approval thresholds $\alpha_1, \ldots, \alpha_k$, in conjunction with the other types of information that we use here, and quantify the effect on the distortion. More generally, it would make sense to consider different types of information structures and combinations between them; note, for example, that our $\alpha$-\textsc{Most-Compact-Set} mechanism for the maximum cost uses the distances between alternatives, the $\alpha$-TAS, and information about the top-ranked alternative of each agent. Finally, a very meaningful avenue is to consider how randomization might aid in achieving even more improved distortion bounds, in the presence of the $\alpha$-TAS. 

\bibliographystyle{plainnat}
\bibliography{references}

\end{document}